\documentclass[aps,pra,groupedaddress,amsmath,amssymb,tightenlines,notitlepage,nofootinbib,twocolumn,superscriptaddress]{revtex4-1}
\pdfoutput=1

\PassOptionsToPackage{pdftitle={Overcoming entropic limitations on asymptotic state transformations through probabilistic protocols}}{hyperref}

\usepackage{newpxtext,newpxmath}

\usepackage{bm,bbm}% 
\usepackage{braket}
\usepackage{amsthm,amsmath,amssymb}

\usepackage[utf8]{inputenc}

\usepackage[dvipsnames,svgnames]{xcolor}
\usepackage{amsmath,amssymb}
\usepackage{changepage,afterpage}
\usepackage{etruscan}
\usepackage{tikz}
\usepackage{enumerate}
\usepackage{mathtools}
\usepackage[colorlinks=true, allcolors=MidnightBlue!70!black!70!TealBlue]{hyperref}

\usepackage{array,booktabs,multirow}
\usepackage{siunitx}

\usepackage{xltabular}

\newtheorem{theorem}{Theorem}

\newtheorem{proposition}[theorem]{Proposition}

\newtheorem{lemma}[theorem]{Lemma}
\newtheorem{fact}{Fact}

\newtheorem{stheorem}{Theorem}

\newtheorem{sproposition}[stheorem]{Proposition}
\newtheorem{scorollary}[stheorem]{Corollary}
\newtheorem{slemma}[stheorem]{Lemma}

\theoremstyle{definition}
\newtheorem*{remark}{Remark}

\newtheoremstyle{red}{}{}{\normalfont}{}{\color{red!80!black}\bfseries}{.}{ }{}
\theoremstyle{red}

\let\nc\newcommand

\nc{\NegW}{W_\tau}

\nc{\RW}{\Omega_\FF}

\renewcommand{\*}{\textup{*}}

\renewcommand{\bar}{\;\rule{0pt}{9.5pt}\right|\;}

\newcommand{\lset}{\left\{\left.}
\newcommand{\rset}{\right\}}
\nc{\lsetr}{\left\{\,}
\nc{\rsetr}{\right.\right\}}
\nc{\barr}{\;\rule{0pt}{9.5pt}\left|\;}

\DeclareMathOperator{\Tr}{Tr}
\DeclareMathOperator{\supp}{supp}

\newcommand{\norm}[2]{\left\lVert#1\right\rVert_{\,#2}}
\newcommand{\proj}[1]{\ket{#1}\!\bra{#1}}
\nc{\id}{\mathbbm{1}}

\newcommand{\F}{\mathcal{F}}

\renewcommand{\H}{\mathcal{H}}
\newcommand{\D}{\mathcal{D}}

\renewcommand{\O}{\mathcal{O}}

\newcommand{\RR}{\mathbb{R}}

\newcommand{\DD}{\mathbb{D}}

\newcommand{\OO}{\mathcal{O}}
\newcommand{\FF}{\mathcal{F}}

\nc{\Dmax}{D_{\max}}
\let\F\FF
\let\O\OO

\nc{\SEP}{\mathrm{SEP}}
\nc{\STAB}{\mathrm{STAB}}
\nc{\PPT}{\mathrm{PPT}}
\nc{\PPTP}{\mathrm{PPTP}}
\nc{\SEPP}{\mathrm{SEPP}}
\nc{\SP}{\mathrm{KP}}
\nc{\FP}{\mathrm{FP}}
\nc{\CPTP}{{\mathrm{CPTP}}}
\let\wt\widetilde

\DeclareMathOperator{\cone}{cone}

\nc{\Op}{\mathcal{O}}

\nc{\idc}{\mathrm{id}}
\nc{\ve}{\varepsilon}

\nc{\Omax}{\O_{\mathrm{max}}}

\usepackage{stackengine,moresize}
\stackMath

\let\textleq\relax
\let\textgeq\relax
\let\texteq\relax

\newcommand{\texteq}[1]{\stackrel{\mathclap{\scriptsize \mbox{#1}}}{=}}
\newcommand{\textleq}[1]{\stackrel{\mathclap{\scriptsize \mbox{#1}}}{\leq}}

\newcommand{\textgeq}[1]{\stackrel{\mathclap{\scriptsize \mbox{#1}}}{\geq}}

\newcommand{\floor}[1]{\left\lfloor #1 \right\rfloor}

\nc{\sminfty}{{\infty,\bullet}}

\usepackage[most,breakable]{tcolorbox}
\renewenvironment{boxed}[1]%
  {\expandafter\ifstrequal\expandafter{#1}{orange}{\begin{tcolorbox}[colback=orange!5,colframe=orange!15,breakable,enhanced]}{\begin{tcolorbox}[colback=white,colframe=gray!10,breakable,enhanced]}}%
  {\end{tcolorbox}}

\nc{\regrob}{\DD_{s,\FF}^{\infty}}
\nc{\regrobs}{\DD_{s,\FF}^{\sminfty}}

\makeatletter 
\renewcommand\onecolumngrid{% 
\do@columngrid{one}{\@ne}%
\def\set@footnotewidth{\onecolumngrid}% <<<<<<<<<<<<<<<<
\def\footnoterule{\kern-6pt\hrule width 1.5in\kern6pt}%
}
\makeatother

\usepackage{fmtcount,refcount}

\usepackage[left=.73in,right=.73in,top=.7in,bottom=1in]{geometry}

\begin{document}

\title{Overcoming entropic limitations on asymptotic state transformations\texorpdfstring{\\}{ }through probabilistic protocols}
 
\author{Bartosz Regula}
\email{bartosz.regula@gmail.com}
\affiliation{Department of Physics, Graduate School of Science, The University of Tokyo, Bunkyo-ku, Tokyo 113-0033, Japan}

\author{Ludovico Lami}
\email{ludovico.lami@gmail.com}
\affiliation{Institut f\"{u}r Theoretische Physik und IQST, Universit\"{a}t Ulm, Albert-Einstein-Allee 11, D-89069 Ulm, Germany}

\author{Mark M. Wilde}
\email{wilde@cornell.edu}
\affiliation{School of Electrical and Computer Engineering, Cornell University, Ithaca, New York 14850, USA}
\affiliation{Hearne Institute for Theoretical Physics, Department of Physics and Astronomy, and Center for Computation and Technology, Louisiana State University, Baton Rouge, Louisiana 70803, USA}

%%%%

\begin{abstract}
The quantum relative entropy is known to play a key role in determining the asymptotic convertibility of quantum states in general resource-theoretic settings, often constituting the unique monotone that is relevant in the asymptotic regime.
We show that this is no longer the case when one allows stochastic protocols that may only succeed with some probability, in which case the quantum relative entropy is insufficient to characterize the rates of asymptotic state transformations, and a new entropic quantity based on a regularization of the Hilbert projective metric comes into play. 
Such a scenario is motivated by a setting where the cost associated with transformations of quantum states, typically taken to be the number of copies of a given state, is instead identified with the size of the quantum memory needed to realize the protocol. 
Our approach allows for constructing transformation protocols that achieve strictly higher rates than those imposed by the relative entropy.
Focusing on the task of resource distillation, we give broadly applicable strong converse bounds on the asymptotic rates of probabilistic distillation protocols, and show them to be tight in relevant settings such as entanglement distillation with non-entangling operations. 
This generalizes and extends previously known limitations that only applied to deterministic protocols. Our methods are based on recent results for probabilistic one-shot transformations as well as a new asymptotic equipartition property for the projective relative entropy.
\end{abstract}

\maketitle

%%%%%%%%%%%%%%%%%%%%%%%%%%%%%%%%%%%%%%%%%%%%%%%%%%%%%%%%%%%%%%%%%%
%%%%%%%%%%%%%%%%%%%%%%%%%%%%%%%%%%%%%%%%%%%%%%%%%%%%%%%%%%%%%%%%%%
%%%%%%%%%%%%%%%%%%%%%%%%%%%%%%%%%%%%%%%%%%%%%%%%%%%%%%%%%%%%%%%%%%

\section{Introduction}

Transformations of quantum states underlie the majority of quantum information processing protocols, and understanding their capabilities and limitations is one of the fundamental problems posed in quantum information science. The ultimate form of such transformations is often understood to be the limit of infinitely many independent and identically distributed (i.i.d.) copies of a given quantum state being coherently manipulated. Although a somewhat idealized scenario, such limits often enjoy simplified properties and have found a multitude of uses in the characterization of different quantum phenomena~\cite{hayashi_2006-2,wilde_2017}. One of the appeals of these approaches is that their asymptotic transformation rates are naturally described by various entropic quantities, giving an explicit operational meaning to measures such as the quantum relative entropy~\cite{hiai_1991,ogawa_2000,matsumoto_2010,buscemi_2019,wang_2019} or the regularized relative entropy of entanglement~\cite{bennett_1996-1,brandao_2010}.
There are, however, many assumptions hidden within these standard Shannon-theoretic approaches.

In computing the asymptotic rates of quantum state manipulation, only one quantity is relevant: 
how many copies of a given state $\rho$ need to be produced per copy of a desired state $\omega$ in order to realize the transformation $\rho \to \omega$.
This is conceptually appealing, but arguably not fully indicative of practical restrictions on state manipulation: this approach assumes that we can coherently manipulate any number of copies $\rho^{\otimes n}$, and indeed there is no cost associated with the size of the quantum memory needed to perform such a manipulation. Ideally, the `cost function'  associated with a transformation should take more parameters into account, reflecting also the difficulty in manipulating many quantum states simultaneously to realize multi-copy operations. 

Here we propose a framework motivated by the opposite point of view: instead of taking into account the number of copies of $\rho$ needed for the transformation, let us focus purely on the quantum memory cost --- that is, how many copies of $\rho$ need to be manipulated at once.
The biggest difference between this setting and conventional approaches is that now we want to avoid manipulating too many states at once, but it does not matter \emph{how many times} we do it. Practically, this setting becomes relevant in a situation where generating copies of $\rho$ is much less expensive than storing and processing them. This is in many respects the case today, as the best sources can generate more than $\sim\!10^5$ entangled photon pairs per second per $\SI{}{\mW}$ of power invested~\cite{steinlechner_12, couteau_2018}, while the best quantum processors cannot store and process more than a few tens of qubits. From the conceptual standpoint, our paradigm 
can be compared with algorithmic complexity classes such as \textsf{PSPACE}, which ignore the time needed to evaluate a program, but tightly constrain its space complexity.

This approach allows us to employ repeat-until-success probabilistic transformation protocols --- potentially more powerful than typically employed deterministic ones --- without incurring an additional cost. The limits of the power of such protocols in the asymptotic setting have not been explored yet, which is what this work aims to address.

\subsection{Results}

We develop a technical toolset allowing for an exact characterization of the asymptotic limitations of probabilistic transformations of quantum states. We first introduce general converse bounds that constrain the performance of all probabilistic transformation schemes, valid in general resource-theoretic settings and thus allowing for broad applicability. 
Our bounds can be understood as the ultimate limitations of probabilistic transformations --- no matter how successful one is in 
the effort to stochastically increase the performance of state manipulation protocols, the restrictions revealed here cannot be overcome. 
We show that the bounds can be tight in many relevant cases, and in particular in characterizing the distillation of resources such as quantum entanglement, as well as in all state transformations within the class of affine resource theories (including thermodynamics, coherence, and asymmetry).
Notably, we exhibit probabilistic protocols that achieve rates strictly larger than bounds based on the quantum relative entropy: this shows that standard Shannon-theoretic approaches are insufficient to characterize the limitations of probabilistic transformations, as the latter can achieve performance forbidden by conventional entropic restrictions.
An explicit example of such a behavior is provided by evaluating our bounds exactly for all isotropic states in entanglement theory.

On the technical side, our methods rely on the regularization of a quantum divergence based on the recently introduced projective robustness~\cite{regula_2022}, and one of our main contributions is to develop an asymptotic equipartition property for this quantity.

We begin with a discussion of the setting in Section~\ref{sec:prelim}. Our main results are divided into two parts: general upper bounds on the performance of probabilistic protocols, discussed in Section~\ref{sec:converse}, and the achievability results showing the tightness of those bounds, considered in Section~\ref{sec:achiev}. Explicit examples are then discussed in Section~\ref{sec:examples}. 
For simplicity, the detailed technical proofs are deferred to the \hyperref[app:projective_entr]{Appendix}.

%%%%%%%%%%%%%%%%%%%%%%%%%%%%%%%%%%%%%%%%%%%%%%%%%%%%%%%%%%%%%%%%%%

\section{Preliminaries}\label{sec:prelim}

\subsection{Quantum resources}

Adopting the framework of quantum resource theories~\cite{chitambar_2019}, our task is to transform a given quantum state $\rho$ into another state $\omega$, using some restricted set of free operations $\OO$ allowed within the given physical setting. An important property of any such free operations is that the set of states that can be prepared with the given operations --- the free states $\FF$ --- remains invariant under them, i.e.\ $\sigma \in \F \Rightarrow \Lambda(\sigma) \in \FF$ for any channel $\Lambda \in \O$. We will take this as the only requirement that the free operations need to satisfy, which is often referred to as resource--non-generating operations. This will ensure that the limitations obtained in our work apply to other choices of physically relevant free operations, which are generally subsets of the resource--non-generating ones.

Above, we have implicitly assumed that both the input and output spaces of the map $\Lambda$ have associated sets of free states $\FF \subseteq \D(\H)$, with $\D(\H)$ denoting all density operators acting on the given Hilbert space $\H$. More generally, when discussing transformations between many copies of quantum states, we will assume that each set $\D(\H^{\otimes n})$ has its own associated set of free states $\FF_n$; we use $\FF$ to refer to the whole family $(\FF_n)_n$ for simplicity.
Throughout this work we assume that the underlying Hilbert space has finite dimension.

In order to ensure that the asymptotic quantities encountered throughout the manuscript are well defined, we follow~\cite{brandao_2010} in introducing a set of basic axioms that the given resource theory should satisfy. These assumptions are obeyed in virtually every theory of interest, and for simplicity we assume that all sets $\FF_n$ considered here satisfy them.

\begin{enumerate}[{Axiom }I.]
\item Each $\FF_n$ is convex and closed.
\item There exists a full-rank state $\sigma$ such that $\sigma^{\otimes n} \in \FF_n$ for all $n$.
\item The sets $\FF_n$ are closed under partial trace: if $\sigma \in \F_{n+1}$, then $\Tr_{k} \sigma \in \F_{n}$ for every $k \in \{1,\ldots,n+1\}$.
\item The sets $\FF_n$ are closed under tensor product: if $\sigma \in \F_n$ and $\sigma' \in \F_{m}$, then $\sigma \otimes \sigma' \in \F_{n+m}$.
\end{enumerate}

\begin{figure}[t]
\centering
\hspace{-0.05cm}\includegraphics[width=8.6cm]{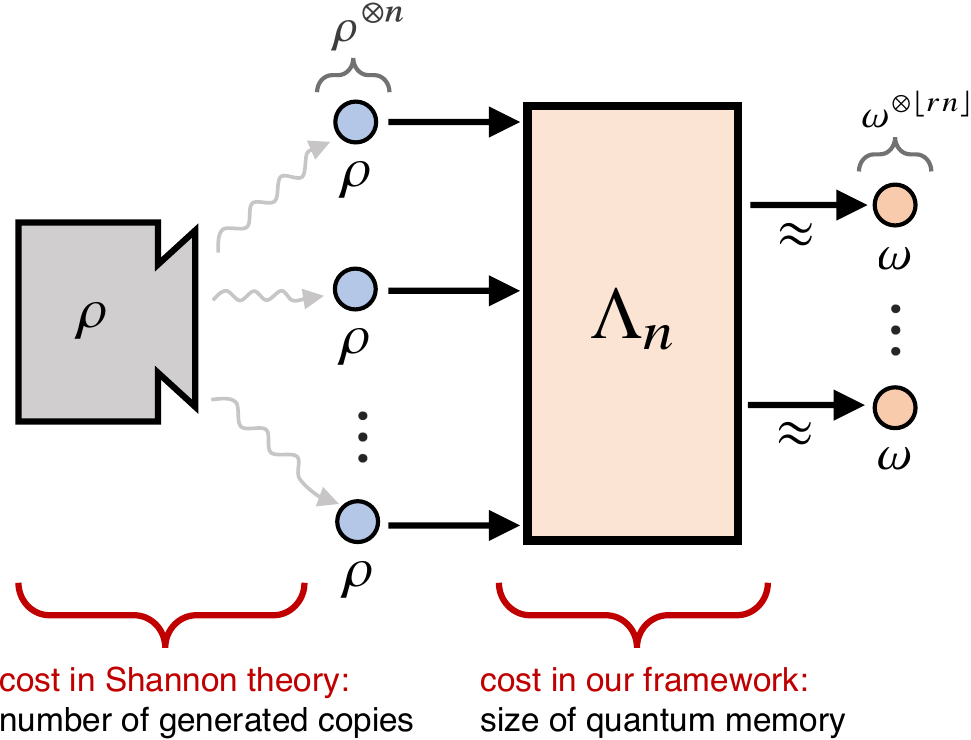}
\caption{\textbf{Asymptotic state transformations.} A general scheme for the conversion of a state $\rho$ into another state $\omega$ takes $n$ copies of $\rho$ and manipulates them with some protocol $\Lambda_n$ such that $\Lambda_n (\rho^{\otimes n}) \approx \omega^{\otimes \floor{rn}}$, with the conversion becoming exact in the limit $n \to \infty$. In conventional quantum Shannon theory, each copy of $\rho$ used incurs a cost; the rate $r$ then tells us how many copies of $\omega$ we can obtain per copy of $\rho$. In our setting, we instead consider the size of the manipulation protocol $\Lambda_n$ to be the costly parameter. If the protocols $\Lambda_n$ were deterministic, the two settings would be exactly the same; however, our approach allows us to take $\Lambda_n$ to be a probabilistic operation and repeat the protocol until it succeeds, at no extra cost.
}
\label{fig:plot}
\end{figure}

\subsection{Asymptotic transformation rates}

A deterministic transformation rate is given by
\begin{equation}\begin{aligned}
  &r(\rho \to \omega) \\&\coloneqq \sup \lsetr r \barr \lim_{n \to \infty}\, \inf_{\Lambda_n \in \OO \cap \mathrm{CPTP}} \norm{\Lambda_n(\rho^{\otimes n}) - \omega^{\otimes \floor{rn}}}{1} = 0 \rsetr,
\end{aligned}\end{equation}
(cf.\ Fig.~\ref{fig:plot}), where we have emphasized that the allowed free operations belong to the set of completely positive and trace-preserving maps (CPTP).
Recall, however, that our setting allows us to employ probabilistic operations --- which are certainly completely positive, but are only required to be trace non-increasing~\cite{davies_1970,ozawa_1984}. 
We will refer to a probabilistic map $\Lambda$ as free (resource non-generating) if it satisfies
\begin{equation}\begin{aligned}
  \sigma\in \FF \quad \Rightarrow \quad \frac{\Lambda(\sigma)}{\Tr \Lambda(\sigma)} \in \FF.
\end{aligned}\end{equation}
Let us then propose an alternative definition of an asymptotic transformation rate as
\begin{equation}\begin{aligned}
  &r_{\mathrm{prob}}(\rho \to \omega) \\&\coloneqq \sup \lsetr r \barr \lim_{n \to \infty} \inf_{\Lambda_n \in \OO} \norm{\frac{\Lambda_n(\rho^{\otimes n})}{\Tr \Lambda_n(\rho^{\otimes n})} - \omega^{\otimes \floor{rn}}}{1} = 0 \rsetr.
\end{aligned}\end{equation}
Just as the conventional rate $r$, the probabilistic rate $r_{\rm prob}$ is defined in the limit $n\to\infty$, but it is ultimately concerned with how many copies of $\omega$ we can obtain per each single copy of $\rho$ that we manipulate.

It is also of interest to study the \emph{strong converse rates} $r^\dagger(\rho \to \omega)$ and $r^\dagger_{\mathrm{prob}}(\rho \to \omega)$, which are defined analogously except that the error, instead of going to $0$, is merely constrained to not tend to $1$. Precisely,
\begin{equation}\begin{aligned}
  &r^\dagger_{\mathrm{prob}}(\rho \to \omega) \\
  &\coloneqq \sup \lsetr r\! \barr \liminf_{n \to \infty} \inf_{\Lambda_n \in \OO   } \frac12 \norm{\frac{\Lambda_n(\rho^{\otimes n})}{\Tr \Lambda_n(\rho^{\otimes n})} - \omega^{\otimes \floor{rn}}}{1} < 1 \rsetr\!,
\end{aligned}\end{equation}
and similarly for $r^\dagger$.
This gives a threshold for achievable protocols: attempting transformations at any rate higher than the strong converse would necessarily incur a very large (tending to $1$) error.

Similar rates were previously studied in the transformations of multipartite entangled pure states~\cite{chitambar_2008,yu_2014,vrana_2017}, in which context a surprising connection with algebraic complexity theory was identified. There, however, no error whatsoever was allowed in the transformation. Although this stricter requirement may be suitable for pure states, such a definition cannot be applied to general quantum systems: in the distillation from noisy mixed states, an asymptotically vanishing error \emph{must} be allowed for the transformations to be possible~\cite{fang_2020,regula_2022}, and --- as we will shortly see --- this error cannot vanish faster than exponentially.

%%%%%%%%%%%%%%%%%%%%%%%%%%%%%%%%%%%%%%%%%%%%%%%%%%%%%%%%%%%%%%%%%%

\subsection{Quantum divergences}

Divergences (relative entropies), typically understood to be entropic distances between density matrices, are a commonly encountered concept in quantum theory~\cite{khatri_2020}. The most fundamental is certainly the quantum relative entropy $D(\rho\|\sigma) = \Tr \rho (\log \rho - \log \sigma)$~\cite{umegaki_1962} itself, where $\rho$ and $\sigma$ are density matrices. In resource-theoretic applications, it becomes important to study the optimized divergence $D_{\F}(\rho) \coloneqq \min_{\sigma \in \FF} D(\rho\|\sigma)$~\cite{vedral_1998, achievability}. Then, due to the potential non-additivity of this function~\cite{vollbrecht_2001}, in asymptotic settings it is the regularized relative entropy~\cite{donald_2002}
\begin{equation}\begin{aligned}
  D_\F^\infty(\rho) \coloneqq \lim_{n \to \infty} \frac1n D_{\F_n}(\rho^{\otimes n})
\end{aligned}\end{equation}
that finds operational applications.

A different divergence, one that finds use primarily in one-shot settings, is the max-relative entropy~\cite{datta_2009} defined as $D_{\max}(\rho\|\sigma) \coloneqq \inf \lset \log\lambda \bar \rho \leq \lambda \sigma \rset$. Defining the optimized max-relative entropy $D_{\max,\FF}$ as above, an important aspect of this quantity is that, after `smoothing' and regularizing, it actually yields the regularized relative entropy itself~\cite{datta_2009-2,brandao_2010}:\footnote{The curious reader might wonder whether this result is related to the generalized quantum Stein's lemma of~\cite{brandao_2010-1}, in whose proof some issues were recently identified~\cite{berta_2022}. Fortunately, the asymptotic equipartition property of $D_{\max}$ that we employ here is independent of that result, as can be seen both in the proof found in~\cite{brandao_2010} and the independent proof in~\cite{datta_2009-2}.}
\begin{equation}\begin{aligned}\label{eq:dmax_regularization}
  \lim_{\ve \to 0}\, \lim_{n\to\infty}\, \min_{\frac12 \norm{\rho'-\rho^{\otimes n}}{1}\leq \ve} \frac{1}{n} D_{\max,\FF_n} (\rho') = D_\F^\infty(\rho),
\end{aligned}\end{equation}
where $\rho'$ is constrained to be a density matrix.

Another well-known one-shot divergence is the min-relative entropy~\cite{datta_2009}, given by $D_{\min}(\psi\|\sigma) = - \log \braket{\psi|\sigma|\psi}$ for a pure state $\psi = \proj{\psi}$.

\section{General limitations on probabilistic transformations}\label{sec:converse}

\subsection{Projective relative entropy}

Our approach will require the study of a different type of divergence, which we dub the \emph{projective relative entropy}:
\begin{equation}\begin{aligned}
  \DD_{\Omega} (\rho \| \sigma) \coloneqq D_{\max}(\rho \| \sigma) + D_{\max} (\sigma \| \rho).
\end{aligned}\end{equation}
This is also known as the Hilbert projective metric between $\rho$ and $\sigma$ with respect to the positive semidefinite cone~\cite{bushell_1973,reeb_2011}. The notation $\DD$ is used here to avoid confusion with quantities based on the standard relative entropy $D$.

The optimized variant of this quantity,
\begin{equation}\begin{aligned}
  \DD_{\Omega,\FF} (\rho) \coloneqq \min_{\sigma \in \FF} \DD_{\Omega} (\rho\|\sigma),
\end{aligned}\end{equation}
which we refer to as the \emph{projective relative entropy of a resource},
is directly related to the projective robustness introduced in~\cite{regula_2022,regula_2021-4} and used to characterize one-shot transformations in probabilistic settings. See Appendix~\ref{app:projective_entr} for a discussion of the properties of $\DD_{\Omega,\FF}$. Since we are interested in asymptotic state manipulation, let us define the regularization
\begin{equation}\begin{aligned}
  \DD_{\Omega,\FF}^\infty (\rho) \coloneqq \lim_{n \to \infty} \frac1n \DD_{\Omega,\FF_n}\!\left(\rho^{\otimes n}\right).
\end{aligned}\end{equation}

As with the max-relative entropy, it is also natural to expect the \emph{smoothed} regularization of $\DD_\Omega$ to come into play; that is, the quantity
\begin{equation}\begin{aligned}
  \DD_{\Omega,\FF}^\sminfty (\rho) \coloneqq \lim_{\ve \to 0}\, \lim_{n\to\infty}\, \min_{\frac12 \norm{\rho'-\rho^{\otimes n}}{1}\leq \ve} \frac{1}{n} \DD_{\Omega,\FF_n} (\rho').
\end{aligned}\end{equation}
We show that this simply equals the regularized relative entropy with respect to the set~$\FF$ (see Appendix~\ref{app:aep}).

\begin{lemma}[Asymptotic equipartition property for the projective relative entropy] \label{lem:omega_reg}
In every convex resource theory with free states $\FF$, it holds that
\begin{equation}\begin{aligned}
  \DD_{\Omega,\FF}^\sminfty (\rho) = D_{\FF}^\infty(\rho).
\end{aligned}\end{equation}
\end{lemma}

That is, although $D_{\max}$ and $\DD_\Omega$ are very differently-behaved quantities, they both give rise to the same quantity asymptotically. %This
The above result will greatly simplify the asymptotic bounds on probabilistic state transformations and allow for direct comparisons with the deterministic case.

%%%%%%%%%%%%%%%%%%%%%%%%%%%%%%%%%%%%%%%%%%%%%%%%%%%%%%%%%%%%%%%%%%

\subsection{Converse bound}

 We now use the regularized projective relative entropy to establish a general converse bound on all state transformations in the probabilistic setting.

\begin{proposition}\label{prop:converse}
For all states $\rho$ and $\omega$, the following  holds:
\begin{equation}\begin{aligned}\label{eq:converse}
  r_{\rm{prob}}(\rho \to \omega) \leq \frac{\DD^\infty_{\Omega,\FF}(\rho)}{D^\infty_{\FF}(\omega)}.
\end{aligned}\end{equation}
\end{proposition}
See Appendix~\ref{app:converse} for  a proof.

This should be compared with the well-known upper bound on \emph{deterministic} transformation rates given by~\cite{horodecki_2001,horodecki_2013-3}
\begin{equation}\begin{aligned}\label{eq:converse_deterministic}
  r(\rho \to \omega) \leq \frac{D^\infty_{\FF}(\rho)}{D^\infty_{\FF}(\omega)}.
\end{aligned}\end{equation}
An interesting aspect is that both bounds have the regularized relative entropy $D^\infty_{\FF}(\omega)$ in the denominator, but the numerator is very different --- the probabilistic bound in \eqref{eq:converse} features the regularization of the projective relative entropy, without any smoothing. It is not difficult to find examples of states such that $\DD^\infty_{\Omega,\FF}(\rho) > D^\infty_{\FF}(\rho)$, meaning that the probabilistic upper bound can be strictly larger. As we will see, this is in fact the best bound possible, as it can be achieved exactly in many cases.

A word of caution is necessary here. It may be the case that the quantity $\DD_{\Omega,\FF}(\rho)$ actually diverges to infinity~\cite{regula_2021-4}, e.g.\ when $\rho$ is a pure state. However, the problem is avoided for the more practically relevant, high-rank quantum states, in which case \eqref{eq:converse} is typically a well-defined and finite bound. 

\subsection{Improved bound for distillation}

We can obtain a number of improvements to our results when the task of \emph{distillation} (purification) is considered. Here, the target state is chosen to be a pure state $\psi$, as is often the case in practical state transformations where one aims to purify a noisy system. We first obtain the following improvement over Proposition~\ref{prop:converse}, the proof of which can be found in Appendix~\ref{app:converse_dist}.

\begin{proposition}\label{prop:dist_bound_rate}
Every sequence $(\Lambda_n)_n$ of distillation protocols satisfies the following trade-off relation between its rate $r$ and transformation errors $\ve_n \coloneqq \frac12 \norm{\Lambda_n(\rho^{\otimes n}) - \psi^{\otimes \floor{rn}}}{1}$:
\begin{equation}\begin{aligned}
  r \leq \frac{\DD_{\Omega,\FF}^\infty(\rho)}{D_{\min,\FF}^\infty(\psi)} - \frac{\limsup_{n\to\infty} \frac1n \log  \ve_n^{-1} }{D_{\min,\FF}^\infty(\psi)}.
\end{aligned}
\label{eq:tradeodd-bnd}\end{equation}
In particular, the strong converse rate satisfies
\begin{equation}\begin{aligned}\label{eq:dist_strong_converse_bound}
  r^{\dagger}_{\mathrm{prob}}(\rho \to \psi) \leq \frac{\DD_{\Omega,\FF}^\infty(\rho)}{D_{\min,\FF}^\infty(\psi)}.
\end{aligned}\end{equation}
\end{proposition}

Here, $D_{\min,\FF}^\infty (\psi) \coloneqq \lim_{n \to \infty} \frac1n D_{\min,\FF_n}\!\left(\psi^{\otimes n}\right)$. The bound in \eqref{eq:tradeodd-bnd} allows one to understand exactly the rates achievable with a given error sequence $(\ve_n)_n$. One important consequence is that, if $\ve_n$ goes to zero too quickly (faster than exponentially), then one \emph{cannot} distill at any non-zero rate; in other words, the result shows that for every distillation protocol, including the most general probabilistic ones, the errors must satisfy~$\ve_n = 2^{-O(n)}$.

%%%%%%%%%%%%%%%%%%%%%%%%%%%%%%%%%%%%%%%%%%%%%%%%%%%%%%%%%%%%%%%%%%

\section{Achievability results}\label{sec:achiev}

\subsection{Affine resource theories}

Our general converse bound gives a universal constraint: no matter how many copies of states we have at our disposal, and no matter how many times we repeat a given protocol, the bound of Proposition~\ref{prop:converse} cannot be exceeded. In order to understand the tightness of this result, it is then of interest to investigate when the bound can be actually achieved, giving an exact expression for the asymptotic transformation rate between any two states. 
This is the case in the class of \emph{affine} resource theories~\cite{gour_2017,regula_2020}, defined as such that the set of free states $\FF$ is the intersection of some affine subspace of Hermitian operators with the set of all density matrices. This class includes, for instance, the theories of thermodynamics (athermality)~\cite{brandao_2013}, coherence~\cite{baumgratz_2014}, asymmetry~\cite{gour_2008}, or imaginarity~\cite{wu_2021}. 

\begin{proposition}\label{prop:achievability}
Consider any affine resource theory. Then, for all states $\rho$ and  $\omega$, the transformation rate under resource--non-generating operations $\OO$ satisfies %$r_{\rm{prob}}(\rho \to \omega) =  \DD^\infty_{\Omega,\FF}(\rho) / D_{\FF}^\infty(\omega)$.
\begin{equation}\begin{aligned}
  r_{\rm{prob}}(\rho \to \omega) = \frac{ \DD^\infty_{\Omega,\FF}(\rho)}{D_{\FF}^\infty(\omega)}.
\end{aligned}\end{equation}
\end{proposition}

See Appendix~\ref{app:achiev_aff} for a proof. This result is more general than known results in the characterization of deterministic state conversion. While the relative entropy upper bound in Eq.~\eqref{eq:converse_deterministic} is tight in some theories (e.g., athermality~\cite{faist_2019} or coherence~\cite{chitambar_2018}), there is no known result which shows that bound to be tight in general classes of resources.\footnote{We remark that the work~\cite{brandao_2015} claimed to show that the relative entropy bound is tight in almost all resource theories. However, due to issues in the proof of the underlying work~\cite{brandao_2010-1}, the result is not known to be true~\cite{berta_2022}. Furthermore, the framework of~\cite{brandao_2015} employs, instead of free operations $\OO$, a class of operations that is only \emph{approximately} free, and may actually create large amounts of resources in certain cases~\cite{lami_2021-1}.}%~\cite{Note-relent}. 

\subsection{General resource theories}

Going beyond affine resource theories requires a slightly different approach. To this end, consider the \emph{standard robustness} of a given resource~\cite{vidal_1999}, $R_{s,\FF} (\rho) \coloneqq \inf \lset \lambda \bar \rho + \lambda \sigma \propto \sigma' \in \FF,\; \sigma \in \FF \rset$, 
together with its regularized variant
\begin{equation}\begin{aligned}
  \regrobs (\rho) \coloneqq \lim_{\ve \to 0}\, \limsup_{n\to\infty}\, \min_{\frac12 \norm{\rho'-\rho^{\otimes n}}{1}\leq \ve} \frac{1}{n} \log \left( 1 + R_{s,\FF_n}(\rho') \right)\!.
\end{aligned}\end{equation}
Contrary to the case of $D_{\max,\FF}$ in Eq.~\eqref{eq:dmax_regularization}, we do not know how to  
express this quantity with an alternative formula that does not involve a smoothing. Nevertheless, it can be used to establish a general achievable transformation rate.

\begin{proposition}\label{prop:achievability_Rs}
Consider any resource theory such that $R_{s,\FF}(\rho)<\infty$ for all states. Then, for all states $\rho$ and  $\omega$, the transformation rate under resource--non-generating operations $\OO$ satisfies
\begin{equation}
    r_{\rm{prob}}(\rho \to \omega) \geq  \frac{\DD^\infty_{\Omega,\FF}(\rho)}{\regrobs(\omega)}.
\end{equation}
\end{proposition}

Note that the achievable rate in Proposition~\ref{prop:achievability_Rs} does not always match the converse bound in Proposition~\ref{prop:converse}, as recently an example of a state satisfying $\regrobs (\rho) > D^\infty_\F(\rho)$ was found~\cite{lami_2021-1}. However, we will see that the two quantities can be equal in some relevant cases.

%%%%%%%%%%%%%%%%%%%%%%%%%%%%%%%%%%%%%%%%%%%%%%%%%%%%%%%%%%%%%%%%%%

\subsection{Distillation}

The task of distillation is a special case for which, in many theories of interest and even for some non-affine theories, we can evaluate the optimal asymptotic rate exactly because our upper (Proposition~\ref{prop:converse} or~\ref{prop:dist_bound_rate}) and lower (Proposition~\ref{prop:achievability} or~\ref{prop:achievability_Rs}) bounds actually coincide. 
This is because the target state of distillation $\psi$ is often chosen to be a maximally resourceful state (e.g., a Bell state $\Phi_2$ in entanglement theory), whose properties allow for a simplified quantification of its resources.

Notably, we can exactly evaluate the rate of probabilistic entanglement distillation under non-entangling operations from any state:
\begin{equation}\begin{aligned}
  r_{\mathrm{prob}}(\rho \to \Phi_2) = r^\dagger_{\mathrm{prob}}(\rho \to \Phi_2) =  \DD_{\Omega,\SEP}^\infty(\rho),
\end{aligned}\end{equation}
where $\SEP$ denotes the set of separable states. This can be considered surprising in light of the fact that, for deterministic transformations, the exact rate of distillation under non-entangling operations is not known~\cite{brandao_2010,berta_2022}.

A similar result can be derived for more general resource theories, and we present a detailed technical discussion of this property in Appendix~\ref{app:ach_dist}.

\subsection{On the probability of success in the asymptotic limit}

Many conventional distillation protocols also incorporate probabilistic elements --- including even the earliest schemes for entanglement distillation~\cite{bennett_1996-1} --- but in such settings these `probabilistic' rates have been shown to be equivalent to standard deterministic ones~\cite{rains_1999}. What is the difference between such protocols and the transformations considered here?

To gain some insight into why the standard relative entropy $D^\infty_{\FF}(\rho)$ is insufficient to characterize the asymptotic rates of probabilistic transformations as considered in our work, let us look into the optimal probability of success in such protocols. By construction, transformation rates in our framework do not depend on the probability of success of the protocols. This means that our definition of a probabilistic rate $r_{\rm prob}$ allows for a sequence of protocols $(\Lambda_n)_n$ such that the probability of success vanishes in the asymptotic limit, that is, $\lim_{n\to\infty} \Tr \Lambda_n(\rho^{\otimes n}) = 0$. This is because we are effectively disregarding the exact probability, so the rate is not affected by whether $\Tr \Lambda_n(\rho^{\otimes n})$ becomes arbitrarily small.

This is precisely the aspect that distinguishes our approach from conventional asymptotic transformations. To see this, we define a variant of a probabilistic rate where protocols with vanishing success probability are not allowed:
\begin{equation}\begin{aligned}
  &r_{\mathrm{prob} > 0}(\rho \to \omega) \\
  &\coloneqq \sup \bigg\{ \,r \;\bigg|\; (\Lambda_n)_n \in \OO,\; \liminf_{n \to \infty}\, \Tr \Lambda_n(\rho^{\otimes n}) > 0,\\
  &\hphantom{\coloneqq \sup \bigg\{ r \;\bigg|\;} \lim_{n \to \infty}  \frac12 \norm{\frac{\Lambda_n(\rho^{\otimes n})}{\Tr \Lambda_n(\rho^{\otimes n})} - \omega^{\otimes \floor{rn}}}{1} = 0 \,\bigg\}.
\end{aligned}\end{equation}
Our asymptotic equipartition property of Lemma~\ref{lem:omega_reg_full} then allows us to show that the deterministic upper bound based on the relative entropy $D_\FF^\infty$ cannot be exceeded by such protocols.

\begin{proposition}\label{prop:vanishing_success}
The rate of any transformation with a probability of success that does not asymptotically vanish satisfies
\begin{equation}\begin{aligned}\label{eq:vanishing_success}
  r_{\rm{prob} >0}(\rho \to \omega) \leq \frac{D^\infty_\FF(\rho)}{D_\FF^\infty(\omega)}.
\end{aligned}\end{equation}
\end{proposition}

This result shows that a vanishing probability of success is, in many cases, \emph{required} to gain an advantage over probabilistic protocols in the asymptotic limit. That is because the rate $\frac{D^\infty_\FF(\rho)}{D_\FF^\infty(\omega)}$ is achievable deterministically in many known settings, and indeed it has been conjectured to be achievable in general resource theories under so-called asymptotically resource--non-generating operations~\cite{brandao_2010,brandao_2015,berta_2022}. Whenever this is the case, we have that the two rates are equal, $r_{\rm{prob} >0}(\rho \to \omega) = r(\rho \to \omega)$, while the rate $r_{\rm{prob}}(\rho \to \omega)$ considered in our work is typically strictly larger (see e.g.\ the upcoming Section~\ref{sec:iso}).

Remarkably, standard techniques leveraging the asymptotic continuity and `not-too-convexity' of $D_\FF$~\cite{winter_2016} do not suffice to prove Proposition~\ref{prop:vanishing_success} --- they only succeed in establishing the weaker statement that the \emph{deterministic} rate is upper bounded by the right-hand side of~\eqref{eq:vanishing_success}. Our proof, instead, makes essential use of the asymptotic equipartition property for $D_{\Omega,\FF}$ established in Lemma~\ref{lem:omega_reg}. For further details on this point, see Appendix~\ref{sec:app_vanishing}.

%%%%%%%%%%%%%%%%%%%%%%%%%%%%%%%%%%%%%%%%%%%%%%%%%%%%%%%%%%%%%%%%%%

\begin{figure}[t]
\centering
\hspace{-0.05cm}\includegraphics[width=8cm]{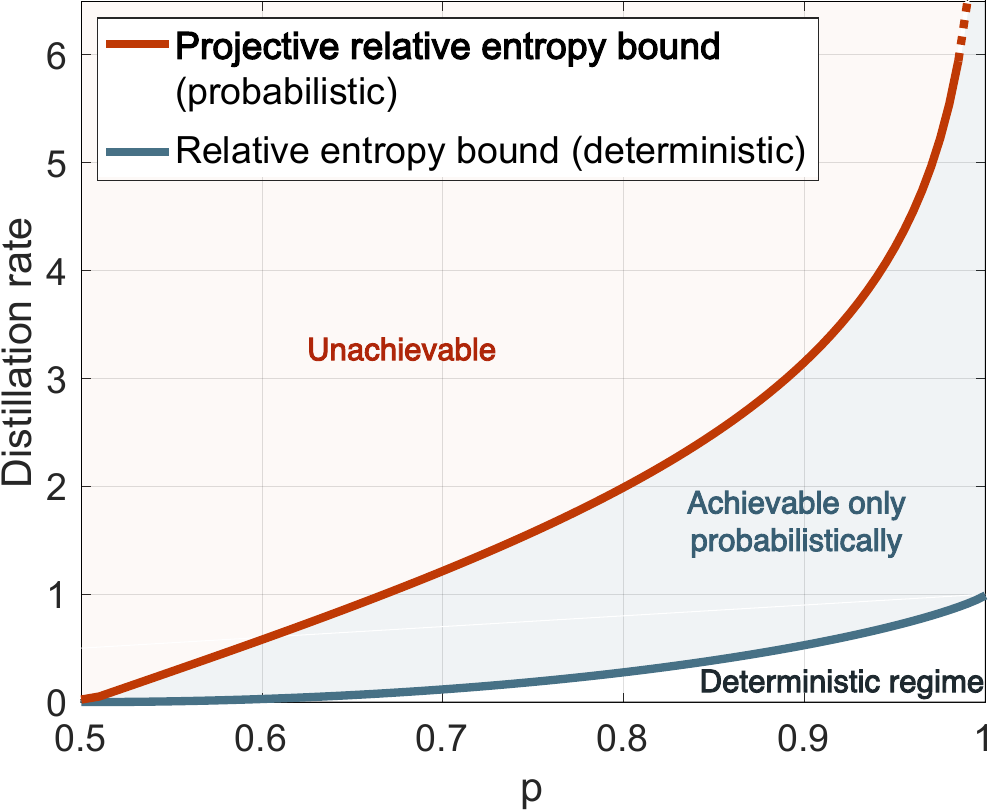}
\caption{\textbf{Entanglement distillation from two-qubit isotropic states.} We plot the most general upper bound on the rate of deterministic entanglement distillation under all non-entangling protocols, the regularized relative entropy $D_\SEP^\infty(\rho_p)$, and compare it with the exact achievable rate of probabilistic entanglement distillation under non-entangling protocols, namely, the regularised projective relative entropy $\DD_{\Omega,\SEP}^\infty(\rho_p)$. It can be seen that, for all values of $p > 0.5$, allowing non-deterministic protocols leads to significantly higher distillation rates. We also observe that the probabilistic rate becomes unbounded as $p \to 1$, which is consistent with the fact that all pure states can be probabilistically interconverted with non-entangling operations~\cite{contreras-tejada_2019,regula_2021-4}. The local dimension is chosen to be $d=2$.
}
\label{fig:rates}
\end{figure}

\section{Examples}\label{sec:examples}

\subsection{Isotropic entangled states}\label{sec:iso}

Isotropic states are representative examples of entangled states that enjoy simplified entanglement properties due to their strong symmetry~\cite{horodecki_1999-1}.
For a local dimension $d$, they are defined as
\begin{equation}\begin{aligned}
  \rho_p \coloneqq p \Phi_d + (1-p) \frac{\id - \Phi_d}{d^2-1}
\end{aligned}\end{equation}
where $\Phi_d$ is a two-qudit maximally entangled state.
The set of free states in this resource theory is given by all separable states, $\F = \SEP$. (One could alternatively consider the resource theory where $\F$ is given by all states with a positive partial transpose; the results below remain unchanged.)

We can use our results to exactly evaluate the rate of probabilistic entanglement distillation: %$r_{\rm prob}(\rho_p \to \Phi_2)$
assuming that $p\geq 1/d$ (as otherwise $\rho_p$ is separable~\cite{horodecki_1999-1}), we find that
\begin{equation*}\begin{aligned}
  r_{\rm prob}(\rho_p \!\to\! \Phi_2) = \DD^\infty_{\Omega,\SEP} (\rho_p) = \DD_{\Omega,\SEP} (\rho_p) = \log \frac{p(d\! -\! 1)}{1\! -\! p},
\end{aligned}\end{equation*}
while any rate achievable under deterministic protocols satisfies~\cite{hayashi_2006-2,brandao_2010,rains_1999-1}
\begin{equation}\begin{aligned}
% BIID
  r^{\dagger}(\rho_p\! \!\to\! \Phi_2) &\leq D^\infty_{\SEP} (\rho_p) \\
  & = p \log d + (1 - p) \log \frac{d}{d\!-\!1} -  h_2(p),
\end{aligned}\end{equation}
where $h_2$ is the binary entropy function, and this bound is conjectured to be tight~\cite{brandao_2010,berta_2022}. The gap between the two quantities is depicted in Fig.~\ref{fig:rates}, showing that probabilistic asymptotic protocols exhibit prominently higher rates than purely deterministic ones.

%%%%%%%%%%%%%%%%%%%%%%%%%%%%%%%%%%%%%%%%%%%%%%%%%%%%%%%%%%%%%%%%%%

\subsection{Dichotomies and distinguishability}

The case of transformations of pairs of quantum states, i.e.\ finding a channel which satisfies $\Lambda(\rho_1) = \omega_1$ and $\Lambda(\rho_2)=\omega_2$, has been studied under the name of quantum dichotomies~\cite{buscemi_2019}  and underlies the resource theory of asymmetric distinguishability~\cite{matsumoto_2010,wang_2019}. Asymptotic transformations here are studied in the sense that the transformation $\rho_1^{\otimes n} \to \omega_1^{\otimes \floor{rn}}$ may be realized approximately (as in the definition of the asymptotic rate), but $\rho_2^{\otimes n} \to \omega_2^{\otimes \floor{rn}}$ must always be exact. It was then shown that the deterministic rate of such transformations is given by $D(\rho_1\|\rho_2) / D(\omega_1\|\omega_2)$~\cite{buscemi_2019,wang_2019}. We get an analogous result also in our case: the probabilistic conversion rate is exactly $\DD_\Omega(\rho_1 \| \rho_2) / D(\omega_1 \| \omega_2)$, which is typically strictly larger than the deterministic rate. A special case of this task has been further studied in~\cite{regula_2022-4} and interpreted therein as a postselected variant of quantum hypothesis testing.

This result can also be applied to resource theories with only a single free state, such as the resource theory of athermality with Gibbs-preserving operations or the resource theory of purity~\cite{gour_2015}.

%%%%%%%%%%%%%%%%%%%%%%%%%%%%%%%%%%%%%%%%%%%%%%%%%%%%%%%%%%%%%%%%%%

\section{Discussion}

We introduced a framework for the asymptotic manipulation of quantum states with probabilistic protocols and established comprehensive methods for its characterization. We specifically used a class of resource monotones based on the projective relative entropy $\DD_{\Omega}$ to establish general upper and lower bounds for transformation rates, showing them to be tight and exactly computable in relevant cases.

There are two main facts revealed by the results of our work. \textit{A priori}, it is not even obvious if probabilistic rates are constrained at all, and even if so, what these constraints may be --- if we are allowing one to repeat the transformation protocols an unbounded number of times, could it not be feasible that \emph{every} transformation becomes eventually possible? We showed that not to be the case, revealing general limitations that every probabilistic state transformation protocol is subject to. 
On the other hand, we saw that probabilistic protocols can (and generally do) outperform deterministic ones, showing that the framework considered in our work does exceed the capabilities of standard quantum Shannon theory.

%%%%%%%%%%%%%%%%%%%%%%%%%%%%%%%%%%%%%%%%%%%%%%%%%%%%%%%%%%%%%%%%%%

%
\begin{acknowledgments}

%\vspace*{.4\baselineskip}

We thank Harry Buhrman for making us aware of a connection between probabilistic transformation rates and algebraic complexity theory. B.R.\ was supported by the Japan Society for the Promotion of Science (JSPS) KAKENHI Grant No.\ 21F21015 and the JSPS Postdoctoral Fellowship for Research in Japan. L.L.\ acknowledges support from the Alexander von Humboldt Foundation. M.M.W.~acknowledges support from the NSF under grant no.~1907615.  We thank the Hearne Institute for Theoretical Physics at Louisiana State University for welcoming B.R.\ and L.L.\ during the visit that led to this work. B.R.\ also gratefully acknowledges the hospitality of the Institute of Theoretical Physics at Ulm University.
\end{acknowledgments}

\clearpage
\appendix
\newgeometry{left=1.2in,right=1.2in,top=.7in,bottom=1in}

\onecolumngrid

%\vspace*{.4\baselineskip}
\setlength{\extrarowheight}{4pt}
\begin{table}[h]
\hspace*{-50pt}
\begin{tabular}{@{}p{2cm}p{8cm}p{6cm}p{2.5cm}@{}}
\toprule
%%%%%%%%%%%%%%%%%%%%%%%%%%%%%%%%%%%%%%%%%%%%%%%%%%%%%%%%%%%%%%%%%%%%%%%%%%%%%%%%%%%%%%%%%%%
 \textbf{Quantity} & \textbf{Name} &  \textbf{Definition}     & Remarks  \\[4pt] \toprule\\[-30pt]
%%%%%%%%%%%%%%%%%%%%%%%%%%%%%%%%%%%%%%%%%%%%%%%%%%%%%%%%%%%%%%%%%%%%%%%%%%%%%%%%%%%%%%%%%%%
$D_\FF(\rho)$ & relative entropy of resource & $\min_{\sigma \in \FF} \,D(\rho\|\sigma) $\newline  \hspace*{-8pt}$=\min_{\sigma \in \FF} \,\Tr \rho (\log\rho-\log\sigma)$ 
& \\ \colrule
%%%%%%%%%%%%%%%%%%%%%%%%%%%%%%%%%%%%%%%%%%%%%%%%%%%%%%%%%%%%%%%%%%%%%%%%%%%%%%%%%%%%%%%%%%%
$D_{\max,\FF}(\rho)$ & max-relative entropy of resource & $\min_{\sigma \in \FF} \, D_{\max}(\rho\|\sigma)$ \newline  \hspace*{-10pt} $= \min_{\sigma \in \FF,\, \lambda \in \RR_+} \lset \log \lambda \bar \rho \leq \lambda \sigma \rset$ \\[4pt] \colrule
%%%%%%%%%%%%%%%%%%%%%%%%%%%%%%%%%%%%%%%%%%%%%%%%%%%%%%%%%%%%%%%%%%%%%%%%%%%%%%%%%%%%%%%%%%%
$D_{\min,\FF}(\rho)$ & min-relative entropy of resource & $\min_{\sigma \in \FF} \, D_{\min}(\rho\|\sigma)$ \newline  \hspace*{-10pt} $= \min_{\sigma \in \FF} %\log \left( \Tr \Pi_\rho \sigma \right)^{-1}
\left[ -\log \left(\Tr \Pi_\rho \sigma\right) \right]$  \\[4pt] \colrule
%%%%%%%%%%%%%%%%%%%%%%%%%%%%%%%%%%%%%%%%%%%%%%%%%%%%%%%%%%%%%%%%%%%%%%%%%%%%%%%%%%%%%%%%%%%
$\DD_{\Omega,\FF}(\rho)$ & projective relative entropy of resource & $\min_{\sigma \in \FF} \, \DD_{\Omega}(\rho\|\sigma)$ \newline  \hspace*{-10pt} $=\min_{\sigma \in \FF} \big[ D_{\max}(\rho\|\sigma) + D_{\max}(\sigma \| \rho) \big]$    \\[1pt] \colrule
%%%%%%%%%%%%%%%%%%%%%%%%%%%%%%%%%%%%%%%%%%%%%%%%%%%%%%%%%%%%%%%%%%%%%%%%%%%%%%%%%%%%%%%%%%%
$\DD_{s,\FF}(\rho)$  & logarithmic standard robustness of resource & $\log \left( 1 + R_{s,\FF} (\rho) \right)$ \newline  \hspace*{-9pt} $= \inf_{\sigma, \sigma' \in \FF,\,\lambda \in \RR_+} \lset \log \lambda \bar \rho = \lambda \sigma - (\lambda - 1) \sigma' \rset$ \\[2pt] \colrule\\[-24pt]
%%%%%%%%%%%%%%%%%%%%%%%%%%%%%%%%%%%%%%%%%%%%%%%%%%%%%%%%%%%%%%%%%%%%%%%%%%%%%%%%%%%%%%%%%%%
$D^\infty_\FF(\rho)$ & regularised relative entropy of resource & $\displaystyle\lim_{n\to\infty} \frac1n D_{\FF_n}\left(\rho^{\otimes n}\right)$ & \\[-4pt] \colrule\\[-24pt]
%%%%%%%%%%%%%%%%%%%%%%%%%%%%%%%%%%%%%%%%%%%%%%%%%%%%%%%%%%%%%%%%%%%%%%%%%%%%%%%%%%%%%%%%%%%
$\DD^\infty_{\Omega,\FF}(\rho)$ & regularised projective relative entropy of resource & $\displaystyle \lim_{n\to\infty} \frac1n \DD_{\Omega,\FF_n}\left(\rho^{\otimes n}\right)$ & \\[-5pt] \colrule\\[-24pt]
%%%%%%%%%%%%%%%%%%%%%%%%%%%%%%%%%%%%%%%%%%%%%%%%%%%%%%%%%%%%%%%%%%%%%%%%%%%%%%%%%%%%%%%%%%%
$D_{\max,\FF}^{\sminfty}(\rho)$ & smoothed regularised max-relative entropy &  $\displaystyle \lim_{\ve \to 0}\, \lim_{n\to\infty}\, \min_{\frac12 \norm{\rho'-\rho^{\otimes n}}{1}\leq \ve} \,\frac{1}{n} \, D_{\max,\FF_n} (\rho')$  & $= D_\F^\infty(\rho)\,$~\cite{datta_2009-2,brandao_2010} \\[10pt]\colrule\\[-24pt]
%%%%%%%%%%%%%%%%%%%%%%%%%%%%%%%%%%%%%%%%%%%%%%%%%%%%%%%%%%%%%%%%%%%%%%%%%%%%%%%%%%%%%%%%%%%
$\DD_{\Omega,\FF}^{\sminfty}(\rho)$ & smoothed regularised projective relative entropy &  $\displaystyle \lim_{\ve \to 0}\, \lim_{n\to\infty}\, \min_{\frac12 \norm{\rho'-\rho^{\otimes n}}{1}\leq \ve} \,\frac{1}{n}\, \DD_{\Omega,\FF_n} (\rho')$  & $= D_\F^\infty(\rho)\,$ [Lem.~\ref{lem:omega_reg}] \\[10pt]\colrule\\[-24pt]
%%%%%%%%%%%%%%%%%%%%%%%%%%%%%%%%%%%%%%%%%%%%%%%%%%%%%%%%%%%%%%%%%%%%%%%%%%%%%%%%%%%%%%%%%%%
$\DD_{s,\FF}^{\sminfty}(\rho)$ & smoothed regularised standard robustness &  $\displaystyle \lim_{\ve \to 0}\, \limsup_{n\to\infty}\, \min_{\frac12 \norm{\rho'-\rho^{\otimes n}}{1}\leq \ve} \,\frac{1}{n} \,\DD_{s,\FF_n} (\rho')$  & $\neq {D}_\F^\infty(\rho)$\newline in general~\cite{lami_2021-1} \\\bottomrule
%%%%%%%%%%%%%%%%%%%%%%%%%%%%%%%%%%%%%%%%%%%%%%%%%%%%%%%%%%%%%%%%%%%%%%%%%%%%%%%%%%%%%%%%%%%
\end{tabular}

\caption{\textbf{Summary of the definitions of the different divergences studied in this work.} We remark that the minima over $\FF$ are achieved due to the lower semicontinuity of the involved quantities and the assumed compactness of $\FF$. The exception is $\DD_{s,\FF}$, which may be unbounded depending on the resource theory. Similarly, lower semicontinuity ensures that the minimisation over the $\ve$-ball in the definitions of $D_{\max,\FF}^{\sminfty}$, $\DD_{\Omega,\FF}^{\sminfty}$, and $\DD_{s,\FF}^{\sminfty}$ is achieved in all cases.
}
\label{table}
\end{table}

%\twocolumngrid
%%%%%%%%%%%%%%%%%%%%%%%%%%%%%%%%%%%%%%%%%%%%%%%%%%%%%%%%%%%%%%%%%%
%\clearpage

%%%%%%%%%%%%%%%%%%%%%%%%%%%%%%%%%%%%%%%%%%%%%%%%%%%%%%%%%%%%%%%%%%
%%%%%%%%%%%%%%%%%%%%%%%%%%%%%%%%%%%%%%%%%%%%%%%%%%%%%%%%%%%%%%%%%%
%%%%%%%%%%%%%%%%%%%%%%%%%%%%%%%%%%%%%%%%%%%%%%%%%%%%%%%%%%%%%%%%%%
%%%%%%%%%%%%%%%%%%%%%%%%%%%%%%%%%%%%%%%%%%%%%%%%%%%%%%%%%%%%%%%%%%
%%%%%%%%%%%%%%%%%%%%%%%%%%%%%%%%%%%%%%%%%%%%%%%%%%%%%%%%%%%%%%%%%%
%%%%%%%%%%%%%%%%%%%%%%%%%%%%%%%%%%%%%%%%%%%%%%%%%%%%%%%%%%%%%%%%%%
%%%%%%%%%%%%%%%%%%%%%%%%%%%%%%%%%%%%%%%%%%%%%%%%%%%%%%%%%%%%%%%%%%

\renewcommand{\thestheorem}{A\arabic{stheorem}}

\begin{center}
\vspace*{\baselineskip}
{\textbf{\large APPENDIX}} \\
\end{center}
%%%%%%%%%%%%%%%%%%%%%%%%%%%%%%%%%%%%%%%%%%%%%%%%%%%%%%%%%%%%%

\section{Notation and basic properties}
\label{app:projective_entr}

\subsection{Projective relative entropy}

In Table~\ref{table} we review the definitions of the various divergences and their regularised forms.

We will sometimes use the notation $\Omega_\FF$ for the non-logarithmic variant of the projective relative entropy $\DD_{\Omega,\FF}$, namely $\Omega_\FF (\rho) = 2^{\DD_{\Omega,\FF}(\rho)}$.

 We recall some of the most important properties of $\DD_{\Omega,\FF}$~\cite{regula_2021-4}:
\begin{enumerate}[(i)]
\item It is monotonic under any probabilistic transformation protocol: for all $\Lambda \in \OO$, $\DD_{\Omega,\FF}\left(\frac{\Lambda(\rho)}{\Tr \Lambda(\rho)}\right)\leq \DD_{\Omega,\FF}(\rho)$.
\item It is faithful, i.e.\ $\DD_{\Omega,\FF}(\rho) = 0$ iff $\rho \in \FF$.
\item It is invariant under scaling, i.e.\ $\DD_{\Omega,\FF}(\lambda \rho) = \DD_{\Omega,\FF}(\rho)$ for all $\lambda > 0$.
\item It may diverge: 
It is finite, i.e., $\DD_{\Omega,\FF}(\rho) < \infty$, if and only if there exists a free state $\sigma \in \FF$ such that $\supp \rho = \supp \sigma$.
\item It is lower semicontinuous, quasiconvex\footnote{We remind the reader that a function $f:C\to \mathbb{R}\cup\{+\infty\}$ defined on a convex set $C$ is said to be quasiconvex if $f\left(tx+(1-t)y\right) \leq \max\{f(x),f(y)\}$ holds for all $x,y\in C$ and $t\in [0,1]$.}, and sub-additive under tensor products.
\item It can be written as a convex optimisation problem:
\begin{equation}\begin{aligned}
  \Omega_\FF(\rho)  &= \inf \lset \gamma \in \RR \bar \rho \leq \wt\sigma \leq \gamma \rho,\; \wt\sigma \in \cone(\FF) \rset\\
  &= \sup \lset \Tr A \rho \bar \Tr B \rho = 1 ,\; B - A \in \cone(\FF)\*,\; A, B \geq 0 \rset\\
  &= \sup \lset \frac{\Tr A \rho}{\Tr B \rho} \bar \frac{\Tr A \sigma}{\Tr B \sigma} \leq 1 \; \forall \sigma \in \FF,\; A, B \geq 0\rset,
\end{aligned}\end{equation}
where $\cone(\FF)$ denotes the convex cone induced by the set of free states $\FF$ and $\cone(\FF)\*$ is its dual cone~\cite{boyd_2004}.
\end{enumerate}

Let us briefly comment on property (iv), that is, the fact that $\DD_{\Omega,\FF}$ might take infinite values. This potential issue is always avoided for highly mixed states --- indeed, Axiom II ensures that there exists a full-rank free state, so $\DD_{\Omega,\FF}(\rho) < \infty$ for all full-rank $\rho$. Extended variants of this property can be shown in specific resource theories; for example, in the theory of entanglement, we can show the following.
\begin{slemma}
Consider the resource theory of bipartite entanglement with local systems of dimensions $d_A$ and $d_B$. Then, every state $\rho$ whose rank satisfies $\operatorname{rank}(\rho) \geq d_A d_B - 1$ is such that $\DD_{\Omega,\SEP}(\rho) < \infty$.
\end{slemma}
\begin{proof}
The full-rank case follows directly from Axiom II, so assume that $\operatorname{rank}(\rho) = d_A d_B - 1$. 
The projection onto the support of $\rho$ is then given by $\Pi_\rho = \id - \proj{\phi}$ for some pure $\ket\phi$. Every such $\Pi_\rho$ is known to be a separable operator~\cite{gurvits_2002}. But $\lambda_{\min}(\rho) \Pi_\rho \leq \rho \leq \lambda_{\max}(\rho) \Pi_\rho$ where $\lambda_{\min/\max}$ denotes the smallest/largest non-zero eigenvalue, meaning that $\DD_{\Omega,\SEP}(\rho) \leq \log \lambda_{\max}(\rho) - \log \lambda_{\min}(\rho) < \infty$.
\end{proof}

For lower-rank states, it appears difficult to obtain general statements, and the verification of the finiteness of $\DD_{\Omega,\FF}$ needs to be performed on a state-by-state basis. Nevertheless, the results of this work can be readily applied to any state that can be shown to have a finite value of $\DD_{\Omega,\FF}$.

Note also that Axioms III-IV guarantee that $\DD_{\Omega,\FF}(\rho) < \infty \iff \DD_{\Omega,\FF}^\infty(\rho) < \infty$.

\subsection{Regularisation}

We often use some well-known facts about asymptotic regularisations of functions, which we collect below.

\begin{fact}[Fekete's lemma~\cite{fekete_1923}]
  Let $f$ be a weakly subaditive function of density matrices, that is, one such that $f(\rho^{\otimes m} \otimes \rho^{\otimes n}) \leq f(\rho^{\otimes m}) + f(\rho^{\otimes n})$ for all states $\rho$ and for all $m,n \in \mathbb{N}$. Then, the regularisation $f^\infty(\rho) \coloneqq \lim_{n\to\infty} \frac1n f(\rho^{\otimes n})$
  exists and equals $\inf_{n} \frac1n f(\rho^{\otimes n})$.
\end{fact}

\begin{fact}[Weak additivity of regularisation~\cite{donald_2002}]
  Let $f$ be a function such that the regularisation $f^\infty$ exists. Then, $f^\infty$ is weakly additive, that is, $f^{\infty}(\rho^{\otimes n}) = n f^\infty (\rho)$.
\end{fact}

%%%%%%%%%%%%%%%%%%%%%%%%%%%%%%%%%%%%%%%%%%%%%%%%%%%%%%%%%%%%%

\section{Asymptotic equipartition property}\label{app:aep}

{%
\renewcommand{\thetheorem}{\getrefnumber{lem:omega_reg}}
\begin{boxed}{white}
\begin{lemma}
\label{lem:omega_reg_full}
For every sequence of sets $(\FF_n)_n$ satisfying Axioms I--IV, the smoothed regularisation of the projective relative entropy is simply the regularised relative entropy of the given resource. That is,
\begin{equation}\begin{aligned}
  \DD_{\Omega,\FF}^\sminfty (\rho) = D_{\FF}^\infty(\rho)
\end{aligned}\end{equation}
where we recall that
\begin{equation}\begin{aligned}
  \DD_{\Omega,\FF}^\sminfty (\rho) = \lim_{\ve \to 0}\, \lim_{n\to\infty}\, \min_{\frac12 \norm{\rho'-\rho^{\otimes n}}{1}\leq \ve} \frac{1}{n} \DD_{\Omega,\FF_n} (\rho')
\end{aligned}\end{equation}
with the minimization over normalised quantum states $\rho'$.
\end{lemma}
\end{boxed}
%\addtocounter{stheorem}{-1}
}%
\begin{proof}
It is already known that, for every sequence of sets $(\FF_n)_n$ satisfying Axioms I--IV, the smoothed regularisation of $D_{\max}$ is precisely the regularised relative entropy $D_{\FF}^\infty$~\cite{datta_2009-2,brandao_2010}%
; specifically,
\begin{equation}\begin{aligned}\label{eq:dmax_reg}
  \lim_{\ve \to 0}\, \lim_{n\to\infty}\, \min_{\frac12 \norm{\rho'-\rho^{\otimes n}}{1}\leq \ve}  \frac1n D_{\max,\FF_n}\left(\rho' \right) = D_{\FF}^\infty(\rho).
\end{aligned}\end{equation}
Since $\DD_{\Omega,\FF}(\rho)\geq D_{\max,\FF}(\rho)$ by definition, this immediately gives that 
\begin{equation}\begin{aligned}
  \DD_{\Omega,\FF}^\sminfty (\rho) \geq D_{\FF}^\infty(\rho),
\end{aligned}\end{equation}
and so it suffices to show the opposite inequality. 
We will proceed to show that the term $D_{\max} (\sigma \| \rho)$ that distinguished $\DD_{\Omega}$ from $D_{\max}$ becomes asymptotically irrelevant, and the two divergences must converge to the same value.

To this end,  
fix $\delta > 0$; for all sufficiently small $\ve > 0$, 
there exists a sequence of states $(\rho'_n)_n$ with $\frac12 \norm{\rho'_n-\rho^{\otimes n}}{1}\leq \frac\ve2$ which satisfies
\begin{equation}\begin{aligned}
  \left| \lim_{n\to\infty}\, \frac1n \min_{\sigma \in \FF_n} D_{\max}\left(\rho'_n \| \sigma \right)  - D_{\FF}^\infty(\rho) \right| \leq \delta.
\end{aligned}\end{equation}
For every such choice of $(\rho'_n)_n$, let $(\sigma_n)_n$ be a sequence of states such that $\min_{\sigma \in \FF_n} D_{\max}\left(\rho'_n \| \sigma \right) = D_{\max}(\rho'_n \| \sigma_n) \eqqcolon \lambda_n \; \forall n$. Define
\begin{equation}\begin{aligned}
  \omega_n \coloneqq \frac{2^{-\eta} \sigma_n + \rho'_n}{1+2^{-\eta}},
\end{aligned}\end{equation}
where we fix $\eta \coloneqq \log(2 \ve^{-1}-1)$, so that
\begin{equation}\begin{aligned}
  \frac12 \norm{\omega_n - \rho'_n}{1} &= \frac12 \norm{ \frac{2^{-\eta} \sigma_n + \rho'_n}{1+2^{-\eta}} - \frac{(1+2^{-\eta}) \rho'_n}{1+2^{-\eta}}}{1}\\
  &= \frac12 \frac{2^{-\eta}}{1+2^{-\eta}} \norm{\sigma_n - \rho'_n}{1}\\
  &\leq \frac{2^{-\eta}}{1+2^{-\eta}}\\
  &= \frac\ve2.
\end{aligned}\end{equation}
By construction, we have that
\begin{equation}\begin{aligned}
  \omega_n &\leq \frac{2^{-\eta} \sigma_n + 2^{D_{\max}(\rho'_n \| \sigma_n)} \sigma_n}{1+2^{-\eta}}\\
  &= \frac{2^{-\eta} + 2^{\lambda_n}}{1+2^{-\eta}} \sigma_n
\end{aligned}\end{equation}
and we additionally observe that
\begin{equation}\begin{aligned}
  \sigma_n &\leq \sigma_n + 2^{\eta} \rho'_n\\
  &= (1+2^{-\eta}) \frac{\sigma_n + 2^{\eta} \rho'_n}{1+2^{-\eta}}\\
  &= (2^\eta+1) \, \omega_n.
\end{aligned}\end{equation}
Thus
\begin{equation}\begin{aligned}
  \frac1n \DD_{\Omega} (\omega_n \| \sigma_n) &\leq \frac1n \log( 2^{-\eta} + 2^{\lambda_n} ) - \frac1n \log( 1+2^{-\eta} ) + \frac1n \log (2^\eta+1)\\
  &\leq \frac1n \log( 1+2^{-\eta} ) + \frac1n \lambda_n - \frac1n \log( 1+2^{-\eta} ) + \frac1n \log (2^\eta+1)\\
  &= \frac1n \lambda_n - \frac1n \log \frac\ve2.
\end{aligned}\end{equation}
Taking the limit then gives
\begin{equation}\begin{aligned}
  \lim_{n\to\infty} \frac1n \DD_{\Omega} (\omega_n \| \sigma_n) \leq \lim_{n\to\infty} \frac1n \lambda_n \leq D_{\FF}^\infty(\rho) + \delta.
\end{aligned}\end{equation}
Altogether we have shown that,
%for all $\delta>0$, there exists a sufficiently small $\ve > 0$ so that 
having fixed $\delta>0$, for all sufficiently small $\ve>0$ we can find a sequence of states $(\omega_n)_n$ such that
\begin{equation}\begin{aligned}
  \frac12\norm{\omega_n - \rho^{\otimes n}}{1} \leq \frac12\norm{\omega_n - \rho'_n}{1} + \frac12\norm{\rho'_n- \rho^{\otimes n}}{1} \leq \ve
\end{aligned}\end{equation}
and
\begin{equation}\begin{aligned}
  \lim_{n\to\infty} \frac1n \min_{\sigma \in \FF_n} \DD_{\Omega}\left(\omega_n \| \sigma \right) \leq D_{\FF}^\infty(\rho) + \delta.
\end{aligned}\end{equation}
%Hence,
Since this holds for all sufficiently small $\ve>0$, we can take the limit and conclude that
\begin{equation}\begin{aligned}
  \lim_{\ve\to0} \lim_{n\to\infty} \min_{\frac12 \norm{\rho'-\rho^{\otimes n}}{1}\leq \ve}  \frac1n \DD_{\Omega,\FF_n}\left(\rho'\right) \leq D_{\FF}^\infty(\rho) + \delta
\end{aligned}\end{equation}
%as was to be shown.
for all $\delta>0$, which implies the claim once one takes $\delta\to 0$.
\end{proof}

A stronger (`strong converse') variant of the above result can be shown for the case when the set $\FF$ consists of a single state, with $\FF_n = \{ \sigma^{\otimes n} \}$. Such a case can be encountered e.g.\ in the transformations of quantum dichotomies or in the resource theories of athermality and purity.

\begin{boxed}{white}
\begin{slemma}\label{lem:omega_reg_singleton}
For all $\ve \in (0,1)$,
\begin{equation}\begin{aligned}
   \lim_{n\to\infty}\, \min_{\frac12 \norm{\rho'-\rho^{\otimes n}}{1}\leq \ve} \frac{1}{n} \DD_{\Omega} (\rho' \| \sigma^{\otimes n}) = D(\rho\|\sigma).
\end{aligned}\end{equation}
\end{slemma}
\end{boxed}
\begin{proof}
The crucial difference now is that
\begin{equation}\begin{aligned}\label{eq:s20}
 \lim_{n\to\infty}\, \min_{\frac12 \norm{\rho'-\rho^{\otimes n}}{1}\leq \ve}  \frac1n D_{\max}\left(\rho' \| \sigma^{\otimes n} \right) = D(\rho\|\sigma)
\end{aligned}\end{equation}
holds for all $\ve\in(0,1)$~\cite{tomamichel_2013,datta_2013-1}, which is directly related to the strong converse property of quantum hypothesis testing. The claim then follows from the simple chain of inequalities:
\begin{equation}\begin{aligned}
D(\rho\|\sigma) &= \lim_{n\to\infty}\, \min_{\frac12 \norm{\rho'-\rho^{\otimes n}}{1}\leq \ve} \frac{1}{n} D_{\max} (\rho' \| \sigma^{\otimes n}) \\
&\leq \lim_{n\to\infty}\, \min_{\frac12 \norm{\rho'-\rho^{\otimes n}}{1}\leq \ve} \frac{1}{n} \DD_{\Omega} (\rho' \| \sigma^{\otimes n}) \\
&\textleq{(i)} \lim_{\ve\to 0} \lim_{n\to\infty}\, \min_{\frac12 \norm{\rho'-\rho^{\otimes n}}{1}\leq \ve} \frac{1}{n} \DD_{\Omega} (\rho' \| \sigma^{\otimes n}) \\
&\texteq{(ii)} D(\rho\|\sigma)\, .
\end{aligned}\end{equation}
Here, (i)~descends from the observation that making $\ve$ smaller can only increase the minimum over the $\ve$-ball, and (ii)~comes from Lemma~\ref{lem:omega_reg_full}.
\end{proof}

%%%%%%%%%%%%%%%%%%%%%%%%%%%%%%%%%%%%%%%%%%%%
\bigskip

\section{General converse}\label{app:converse}

{%
\renewcommand{\thetheorem}{\getrefnumber{prop:converse}}
\begin{boxed}{white}
\begin{proposition}
\label{prop:converse_full}
For all states $\rho$ and $\omega$ such that $D^\infty_{\FF}(\omega) > 0$,
\begin{equation}\begin{aligned}
  r_{\rm{prob}}(\rho \to \omega) \leq \frac{\DD^\infty_{\Omega,\FF}(\rho)}{\DD^\sminfty_{\Omega,\FF}(\omega)} = \frac{\DD^\infty_{\Omega,\FF}(\rho)}{D^\infty_{\FF}(\omega)}.
\end{aligned}\end{equation}
If the set $\FF$ consists of a single state, then the above is actually a strong converse bound; specifically, if $\FF_n = \{\sigma^{\otimes n}\}$ in the input space and $\FF_n = \{\sigma'^{\otimes n}\}$ in the output space, then we get
\begin{equation}\begin{aligned}
  r^\dagger_{\rm{prob}}(\rho \to \omega) \leq \frac{\DD_{\Omega}(\rho \| \sigma)}{D(\omega \| \sigma')}.
\end{aligned}\end{equation}
\end{proposition}
%\addtocounter{stheorem}{-1}
\end{boxed}
}%
\begin{proof}
Suppose that $r$ is an achievable rate; that is, there exists a sequence of protocols $(\Lambda_n)_n \in \OO$ such that $\frac{\Lambda_n(\rho^{\otimes n})}{\Tr \Lambda_n(\rho^{\otimes n})} = \tau_n$ with $\frac12 \norm{\tau_n - \omega^{\otimes \floor{rn}}}{1} \eqqcolon \ve_n \to 0$ as $n \to \infty$. Using the monotonicity of $\DD_{\Omega,\FF}$ under every free probabilistic protocol, we get
\begin{equation}\begin{aligned}
  \DD_{\Omega,\FF_n}\left(\rho^{\otimes n}\right) &\geq  \DD_{\Omega,\FF_{\floor{rn}}}(\tau_n)\\
  &\geq \min_{\frac12 \norm{\omega'-\omega^{\otimes \floor{rn} }}{1}\leq \ve_n} \DD_{\Omega,{\FF_{\floor{rn}}}} \left(\omega'\right).
\end{aligned}\end{equation}
Then
\begin{equation}\begin{aligned}
  \lim_{n \to \infty} \frac1n \DD_{\Omega,\FF}\left(\rho^{\otimes n}\right) &\geq \lim_{n \to \infty} \min_{\frac12 \norm{\omega'-\omega^{\otimes \floor{rn} }}{1}\leq \ve_n} \frac1n \DD_{\Omega,\FF_{\floor{rn}}}\! \left(\omega'\right)\\
  &\geq \lim_{\ve \to 0} \lim_{n \to \infty} \min_{\frac12 \norm{\omega'-\omega^{\otimes \floor{rn} }}{1}\leq \ve} \frac1n \DD_{\Omega,\FF_{\floor{rn}}}\! \left(\omega'\right)\\
  &= r\, D^\infty_\FF(\omega),
\end{aligned}\end{equation}
where in the last line we used Lemma~\ref{lem:omega_reg_full} and the consequent fact that 
\begin{equation}\begin{aligned}\label{eq:weak_additivity_thingy}
  D^\infty_\FF(\omega) &= \lim_{\ve \to 0} \lim_{n\to\infty} \min_{\frac12\norm{\omega' - \omega^{\otimes \floor{rn}}}{1} \leq \ve} \frac{1}{\floor{rn}} \DD_{\Omega,\FF_{\floor{rn}}}\!(\omega')\\
  &= \lim_{\ve \to 0} \lim_{n\to\infty} \frac{n}{\floor{rn}} \min_{\frac12\norm{\omega' - \omega^{\otimes \floor{rn}}}{1} \leq \ve} \frac{1}{n}  \DD_{\Omega,\FF_{\floor{rn}}}\!(\omega')\\
  &= \frac{1}{r} \, \lim_{\ve\to0} \lim_{n\to\infty} \min_{\frac12\norm{\omega' - \omega^{\otimes \floor{rn}}}{1} \leq \ve} \frac{1}{n}  \DD_{\Omega,\FF_{\floor{rn}}}\!(\omega').
\end{aligned}\end{equation}

If $\FF_n = \{\sigma^{\otimes n}\}$, then we no longer need to ensure that $\ve_n \to 0$, as it suffices to take any error $\ve_n \to \ve' < 1$ and invoke Lemma~\ref{lem:omega_reg_singleton}.
\end{proof}

Here we remark that the assumption $D^\infty_{\FF}(\omega) > 0$ might seem to be trivially satisfied, but it is not always the case --- there can indeed exist resource theories where the regularised relative entropy vanishes for some $\omega \notin \FF$, e.g.\ the theory of asymmetry~\cite{gour_2009}. Nevertheless, it is true that $D^\infty_\FF$ can be ensured to be non-zero for all resourceful states in the majority of practically relevant theories, e.g.\ entanglement~\cite{piani_2009-1}, making the assumption always satisfied.

%%%%%%%%%%%%%%%%%%%%%%%%%%%%%%%%%%%%%%%%%%%%%%%%%%%%%%%%%%%%%

\section{Strong converse for distillation}\label{app:converse_dist}

The result below concerns the case when the target state of the transformation is pure.

We will use the $D_{\min}$ relative entropy to the free states, for which we note that
\begin{equation}
    D^\infty_{\min,\FF}(\psi) = \lim_{n\to\infty} \frac1n \log \max_{\sigma \in \FF_n} \braket{\psi^{\otimes n}|\sigma|\psi^{\otimes n}}^{-1}.
\end{equation}

{%
\renewcommand{\thetheorem}{\getrefnumber{prop:dist_bound_rate}}
\begin{boxed}{white}
\begin{proposition}
\label{prop:dist_bound_rate_full}
Consider a pure target state $\omega = \psi \notin \FF$. 
Every physical sequence of distillation protocols $(\Lambda_n)_n$ satisfies the following trade-off relation between its rate $r$ and transformation errors $\ve_n \coloneqq \frac12 \norm{\Lambda_n(\rho^{\otimes n}) - \psi^{\otimes \floor{rn}}}{1}$:
\begin{equation}\begin{aligned}
  r + \frac{\limsup_{n\to\infty} \frac1n \log \left( \ve_n^{-1} - 1 \right)}{D_{\min,\FF}^\infty(\psi)} \leq \frac{\DD_{\Omega,\FF}^\infty(\rho)}{D_{\min,\FF}^\infty(\psi)}.
\end{aligned}\end{equation}

In particular,
\begin{equation}\begin{aligned}
  r^\dagger_{\mathrm{prob}}(\rho \to \psi) \leq \frac{\DD_{\Omega,\FF}^\infty(\rho)}{D_{\min,\FF}^\infty(\psi)}.
\end{aligned}\end{equation}
\end{proposition}
\end{boxed}
%\addtocounter{stheorem}{-1}
}%

Let us remark that, unless the given protocol trivialises with $\lim_{n\to\infty} \ve_n = 1$, then it holds that $\limsup_{n\to\infty} \frac1n \log \left( \ve_n^{-1} - 1 \right) = \limsup_{n\to\infty} \frac1n \log \ve_n^{-1}$.

\begin{proof}
If $\DD^\infty_{\Omega,\FF}(\rho) = \infty$, then the result is trivial, so assume otherwise. 
Assume now that there exists a sequence of protocols $(\Lambda_n)_n \in \OO$ such that $\frac{\Lambda_n(\rho^{\otimes n})}{\Tr \Lambda_n(\rho^{\otimes n})} = \tau_n$ with error $\ve_n \coloneqq \frac12 \norm{\tau_n - \psi^{\otimes \floor{rn} }}{1}$. We use $\delta_n \coloneqq 1-\braket{\psi^{\otimes \floor{rn}} | \tau_n | \psi^{\otimes \floor{rn}}}$ to denote the error in fidelity rather than trace distance.

The result of~\cite[Thm.~9]{regula_2021-4} tells us that for each $n$ we necessarily have that
\begin{equation}\begin{aligned}
  \Omega_{\FF_n}(\rho^{\otimes n}) \geq \frac{ (1-\delta_n) (1- F_{\FF_{\floor{rn}}}(\psi^{\otimes \floor{rn}})) }{ \delta_n\,  F_{\FF_{\floor{rn}}}(\psi^{\otimes \floor{rn}}) },
\end{aligned}\end{equation}
where we have defined
\begin{equation}
    F_\FF(\psi) \coloneqq 2^{-D_{\min,\FF}(\psi)} = \max_{\sigma \in \FF} \braket{\psi|\sigma|\psi}
\end{equation}
for simplicity.
Equivalently, we have that
\begin{equation}\begin{aligned}\label{eq:dist_omega_ineq}
  \left( F_{\FF_{\floor{rn}}}(\psi^{\otimes \floor{rn}})^{-1} - 1 \right) &\leq \left( \delta_n^{-1} - 1 \right)^{-1} \Omega_{\FF_n}(\rho^{\otimes n})\\
&\leq \left( \ve_n^{-1} - 1 \right)^{-1} \Omega_{\FF_n}(\rho^{\otimes n}),
\end{aligned}\end{equation}
where $\delta_n \leq \ve_n$ is a consequence of the tighter Fuchs--van de Graaf inequality $1 - F(\rho, \psi) \leq \frac12 \norm{\rho - \psi}{1}$. 
Taking the logarithm of the above and  dividing by $n$ gives
\begin{equation}\begin{aligned}
  %r \log F_\FF(\psi)^{-1} = 
  \liminf_{n\to\infty} \frac{1}{n} \log \left( F_{\FF_{\floor{rn}}}(\psi^{\otimes \floor{rn}})^{-1} - 1 \right) &\leq \liminf_{n\to\infty} \left[ \frac1n \DD_{\Omega,\FF}(\rho^{\otimes n}) - \frac1n \log \left( \ve_n^{-1} - 1 \right) \right]\\
&\leq \limsup_{n\to\infty} \left[ \frac1n \DD_{\Omega,\FF}(\rho^{\otimes n}) \right] - \limsup_{n\to\infty} \frac1n \log \left( \ve_n^{-1} - 1 \right).
\end{aligned}\end{equation}
Here, in the second line we used the fact that $\liminf_{n\to\infty} (a_n-b_n) \leq \limsup_{n\to\infty} a_n - \limsup_{n\to\infty} b_n$, as one sees immediately by picking a sub-sequence $(a_{n_k}-b_{n_k})_k$ with the property that $\lim_{k\to\infty} b_{n_k} = \limsup_{n\to\infty} b_n$. Now, using the fact that $D^\infty_{\min,\FF}(\psi)$ is %well-defined
well defined (due to the sub-additivity of $D_{\min,\FF}$) and weakly additive, we have that the left-hand side reduces to
\begin{equation}\begin{aligned}
\liminf_{n\to\infty} \frac{1}{n} \log \left( F_{\FF_{\floor{rn}}}(\psi^{\otimes \floor{rn}})^{-1} - 1 \right) &= \lim_{n\to\infty} \frac{1}{n} \log  F_{\FF_{\floor{rn}}}(\psi^{\otimes \floor{rn}})^{-1} \\
&= r \, D_{\min,\FF}^\infty(\psi),
\end{aligned}\end{equation}
which concludes the first part of the proof.

To obtain a strong converse bound, we see that assuming that $\liminf_{n \to\infty} \ve_n \in [0,1)$ entails that $\limsup \ve_n^{-1} > 1$, and hence
\begin{equation}\begin{aligned}
  r \, D_{\min,\FF}^\infty(\psi)\leq \limsup_{n\to\infty} \frac1n \DD_{\Omega,\FF_n}(\rho^{\otimes n})
\end{aligned}\end{equation}
as was to be shown.
\end{proof}

Note that in Proposition~\ref{prop:dist_bound_rate_full} we did not need to assume that $\limsup_{n\to\infty} \frac1n \log \left( \ve_n^{-1} - 1 \right) < \infty$: this is guaranteed by \eqref{eq:dist_omega_ineq} coupled with the submultiplicativity of $\Omega_\FF$, which together ensure that $\ve_n^{-1}$ is upper bounded by 
\begin{equation}\begin{aligned}
\ve_n^{-1} \leq \left[ \Omega_\FF(\rho) \, 2^{-r \, D^\infty_{\min}(\psi)} \right]^n + 1
\end{aligned}\end{equation}
which grows exponentially in $n$. This also leads to a general insight about how fast the errors can decay in general distillation protocols, which we formalise as follows.

\begin{scorollary}
If $\RW(\rho) < \infty$ (in particular for every full-rank $\rho$), then there does not exist any distillation protocol $\rho \to \psi$ such that the error decreases faster than exponentially, even in the probabilistic setting. Specifically, $\ve_n = 2^{-O(n)}$ for every physical distillation protocol.
\end{scorollary}

%%%%%%%%%%%%%%%%%%%%%%%%%%%%%%%%%%%%%%%%%%%%%%%%%%%%%%%%%%%%%

\section{Achievability for affine theories}\label{app:achiev_aff}

Recall that we distinguish two types of resource theories:
\begin{enumerate}[(i)]
\item \emph{Affine resources}, that is, those for which the set of free states $\FF$ contains all states in the affine hull $\operatorname{aff}(\FF)$ (smallest affine subspace that contains $\FF$), i.e.\ $\FF = \D(\H) \cap \operatorname{aff}(\FF)$. Note that every such $\FF$ will have an empty interior as a subset of $\D(\H)$, since for any set with a non-empty interior, $\operatorname{aff}(\FF)$ would simply be the whole space of Hermitian operators.
\item \emph{Full-dimensional resources}, that is, those for which $\FF$ has a non-empty interior as a subset of $\D(\H)$. Equivalently, these are the resources for which $R_{s,\FF}(\rho) < \infty$ for every state.
\end{enumerate}

{%
\renewcommand{\thetheorem}{\getrefnumber{prop:achievability}}
\begin{boxed}{white}
\begin{proposition}
\label{prop:achievability_full}
Consider any affine resource theory satisfying Axioms I--IV. Then, for all states $\rho$ and $\omega$ such that $D_\FF^\infty(\omega) > 0$, the transformation rate under resource--non-generating operations $\OO$ satisfies
\begin{equation}\begin{aligned}
  r_{\rm{prob}}(\rho \to \omega) = \frac{ \DD^\infty_{\Omega,\FF}(\rho)}{D_{\FF}^\infty(\omega)}.
\end{aligned}\end{equation}
When $\FF$ consists of a single state, then
\begin{equation}\begin{aligned}
  r_{\rm{prob}}(\rho \to \omega) = r^\dagger_{\rm{prob}}(\rho \to \omega) = \frac{ \DD_\Omega(\rho \| \sigma)}{D(\omega \| \sigma')}.
\end{aligned}\end{equation}
\end{proposition}
\end{boxed}
%\addtocounter{stheorem}{-1}
}%

\begin{proof}
If $\DD^\infty_{\Omega,\FF}(\rho) = \infty$, then also $\DD_{\Omega,\FF}(\rho) = \infty$, which means that $\rho$ can be converted into \emph{any} other state probabilistically~\cite{regula_2021-4}, making the transformation rate unbounded. 
We can thus assume that $\DD^\infty_{\Omega,\FF}(\rho) < \infty$ in what follows.

Fix an arbitrary $\delta>0$ and consider the rate
\begin{equation}\begin{aligned}
  r = \frac{\DD_{\Omega,\FF}^\infty(\rho)}{D^\infty_\FF(\omega)}\frac{1}{1+\delta}.
\end{aligned}\end{equation}
Let $(\omega_{n})_n$ be a sequence of states such that, for sufficiently small $\ve>0$ and for $n$ large enough, it holds that
\begin{equation}\begin{aligned}
  \frac12 \norm{\omega_n - \omega^{\otimes \floor{r n}}}{1} \leq \ve
\end{aligned}\end{equation}
and
\begin{equation}\begin{aligned}\label{eq:large_enough_n}
\frac{1}{n} \DD_{\Omega,\FF_{\floor{rn}}}(\omega_{n}) &\leq (1+\delta) \lim_{\ve \to 0}\, \lim_{n\to\infty}\, \min_{\frac12\norm{\omega' - \omega^{\otimes \floor{rn}}}{1} \leq \ve}  \frac{1}{n} \DD_{\Omega,\FF_{\floor{rn}}} \left( \omega' \right)\\
&= (1+\delta) \, r D_{\FF}^\infty(\omega)\\
&= \DD_{\Omega,\FF}^\infty(\rho),
\end{aligned}\end{equation}
where in the second line we used Lemma~\ref{lem:omega_reg_full} and Eq.~\eqref{eq:weak_additivity_thingy}, and in the third line we used the definition of $r$.

But since
\begin{equation}\begin{aligned}
  \DD_{\Omega,\FF}^\infty(\rho) \leq \frac1n \DD_{\Omega,\FF_n}(\rho^{\otimes n})
\end{aligned}\end{equation}
%for any $n$ 
due to the subadditivity of $\DD_{\Omega,\FF}$, we have that in fact
\begin{equation}\begin{aligned}\label{eq:sufficient_cond_transf}
  \DD_{\Omega,\FF_{\floor{rn}}}(\omega_{n}) \leq \DD_{\Omega,\FF_n}(\rho^{\otimes n})
\end{aligned}\end{equation}
for all $n$ such that \eqref{eq:large_enough_n} holds. 
As $\DD_{\Omega,\FF}$ is the unique monotone that completely determines the existence of probabilistic transformations in all affine theories~\cite[Theorem~5]{regula_2021-4}, what this entails is that $\rho^{\otimes n}$ can be transformed into $\omega_n$ by a probabilistic resource--non-generating transformation or a sequence thereof; specifically,
\begin{equation}\begin{aligned}\label{eq:sufficient_cond_transf2}
\forall \zeta>0, \; \exists \Lambda \in \OO \;\, \text{s.t.} \;\, \frac12 \norm{\frac{\Lambda(\rho^{\otimes n})}{\Tr\Lambda(\rho^{\otimes n}) } - \omega_n}{1} \leq \zeta.
\end{aligned}\end{equation}

We have thus established the existence of a protocol that, for all sufficiently large $n$, takes $n$ copies of $\rho$ to a state $\omega_n$ that approximates $\omega^{\otimes \floor{rn}}$ arbitrarily closely. Since this holds for every rate $r$ satisfying
\begin{equation}\begin{aligned}
  r < \frac{\DD_{\Omega,\FF}^\infty(\rho)}{D^\infty_\FF(\omega)},
\end{aligned}\end{equation}
taking the supremum over all such rates concludes the proof.
\end{proof}

%%%%%%%%%%%%%%%%%%%%%%%%%%%%%%%%%%%%%%%%%%%%%%%%%%%%%%%%%%%%%

\section{Achievability for non-affine theories}\label{app:achiev_nonaff}

{%
\renewcommand{\thetheorem}{\getrefnumber{prop:achievability_Rs}}
\begin{boxed}{white}
\begin{proposition}
\label{prop:achievability_Rs_full}
Consider any resource theory satisfying Axioms I--IV such that $R_{s,\FF} (\rho) < \infty$ for all states. Then, for all states $\rho$ and $\omega$ such that $\regrobs(\omega) > 0$, the transformation rate under resource--non-generating operations $\OO$ satisfies
\begin{equation}\begin{aligned}
  r_{\rm{prob}}(\rho \to \omega) \geq \frac{ \DD^\infty_{\Omega,\FF}(\rho)}{\regrobs(\omega)},
\end{aligned}\end{equation}
where
\begin{equation}\begin{aligned}
  \regrobs (\omega) \coloneqq \lim_{\ve \to 0}\, \limsup_{n\to\infty}\, \min_{\frac12 \norm{\omega'-\omega^{\otimes n}}{1}\leq \ve} \frac{1}{n} \log \left( 1 + R_{s,\FF} (\omega') \right).
\end{aligned}\end{equation}
\end{proposition}
\end{boxed}
%\addtocounter{stheorem}{-1}
}%

The approach to proving this result will be analogous to the proof of the affine case (Proposition~\ref{prop:achievability_full}). A key step in that proof was the fact that
\begin{equation}\begin{aligned}
  \DD_{\Omega,\FF}(\rho) \geq \DD_{\Omega,\FF}(\omega) \Rightarrow \text{$\rho$ can be converted to $\omega$ probabilistically},
\end{aligned}\end{equation}
which we used in Eq.~\eqref{eq:sufficient_cond_transf}--\eqref{eq:sufficient_cond_transf2}. However, this condition is only valid in affine theories. The corresponding condition in non-affine theories is~\cite[Theorem~7]{regula_2021-4}
\begin{equation}\begin{aligned}
  \DD_{\Omega,\FF}(\rho) \geq \DD_{\Omega,s,\FF}(\omega) \Rightarrow \text{$\rho$ can be converted to $\omega$ probabilistically},
\end{aligned}\end{equation}
where $\DD_{\Omega,s,\FF}$ is a slightly different variant of the projective relative entropy (based on the `free projective robustness' $\Omega_\FF^\FF$~\cite{regula_2022}), defined as
\begin{equation}\begin{aligned}
\DD_{\Omega,s,\FF}(\omega) &\coloneqq \min_{\sigma \in \FF}  \DD_{\Omega,s} (\rho \| \sigma),\\
  \DD_{\Omega,s} (\rho \| \sigma) &\coloneqq \DD_{s} (\rho \| \sigma) + D_{\max}(\sigma \| \rho),\\
  \DD_s (\rho \| \sigma) &\coloneqq \inf \lset \lambda \bar \rho \leq_\FF \lambda \sigma \rset,
\end{aligned}\end{equation}
with $\leq_\FF$ denoting inequality with respect to $\cone(\FF)$, i.e.\ $A \leq_\FF B \iff B-A \in \cone(\FF)$. The main point to note is that
\begin{equation}\begin{aligned}
  \DD_{s,\FF} (\rho) \coloneqq \min_{\sigma\in \FF} \DD_s(\rho\|\sigma) = \log\left(1 + R_{s,\FF} (\rho) \right),
\end{aligned}\end{equation}
which justifies the standard robustness's appearance in Proposition~\ref{prop:achievability_Rs_full}.

The proof then proceeds in two steps, which we state as two lemmas for clarity.
\begin{slemma}\label{lem:Rs1}
Consider the smoothed regularisation of $ \DD_{\Omega,s,\FF}$, namely
\begin{equation}\begin{aligned}
  \DD_{\Omega,s,\FF}^\sminfty (\omega) \coloneqq& \lim_{\ve \to 0}\, \limsup_{n\to\infty}\, \min_{\frac12 \norm{\omega'-\omega^{\otimes n}}{1}\leq \ve} \frac{1}{n} \DD_{\Omega,s,\FF} (\omega').
\end{aligned}\end{equation}
It then holds that
\begin{equation}\begin{aligned}
  \DD_{\Omega,s,\FF}^\sminfty (\omega) &= \lim_{\ve \to 0}\, \limsup_{n\to\infty}\, \min_{\frac12 \norm{\omega'-\omega^{\otimes n}}{1}\leq \ve} \frac{1}{n} \log \left( 1 + R_{s,\FF} (\omega') \right)\\
  &= \DD_{s,\FF}^\sminfty (\omega).
\end{aligned}\end{equation}
\end{slemma}

\begin{proof}
The proof of this statement is completely analogous to that of our Lemma~\ref{lem:omega_reg_full}, with $D_{\max}$ replaced with $\DD_{s}$; for clarity, let us go through the argument in detail.

We start by observing that since $\DD_{\Omega,s}(\rho\|\sigma)\geq \DD_s(\rho\|\sigma)$ by construction, we see that $\DD_{\Omega,s,\FF} (\omega')\geq \DD_{s,\FF}(\omega') = \log\left(1 + R_{s,\FF} (\omega') \right)$ and in turn that $\DD_{\Omega,s,\FF}^\sminfty (\omega)\geq \DD_{s,\FF}^\sminfty (\omega)$. Therefore, it suffices to show the opposite inequality.

Fix $\delta > 0$. For all sufficiently small $\ve > 0$, one can find a sequence of states $(\omega'_n)_n$ with the property that (a)~$\frac12 \norm{\omega'_n-\omega^{\otimes n}}{1}\leq \frac\ve2$, and (b)~there exists a sequence $(\sigma_n)_n$ of states $\sigma_n\in \FF$ such that
\begin{equation}\begin{aligned}
  \left| \lim_{n\to\infty}\, \frac1n \DD_{s}(\omega'_n\|\sigma_n)  - \DD_{s,\FF}^\sminfty (\omega) \right| \leq \delta.
\end{aligned}\end{equation}
For $\eta \coloneqq \log(2 \ve^{-1}-1)$, construct the states
\begin{equation}\begin{aligned}
  \xi_n \coloneqq \frac{2^{-\eta} \sigma_n + \omega'_n}{1+2^{-\eta}} .
\end{aligned}\end{equation}
On the one hand, it holds that
\begin{equation}\begin{aligned}
    \frac12 \norm{\xi_n - \omega'_n}{1} &= \frac12 \norm{ \frac{2^{-\eta} \sigma_n + \omega'_n}{1+2^{-\eta}} - \frac{(1+2^{-\eta}) \omega'_n}{1+2^{-\eta}}}{1}\\
    &= \frac12 \frac{2^{-\eta}}{1+2^{-\eta}} \norm{\sigma_n - \omega'_n}{1}\\
    &\leq \frac{2^{-\eta}}{1+2^{-\eta}}\\
    &= \frac\ve2,
\end{aligned}\end{equation}
entailing that
\begin{equation}\begin{aligned}
    \frac12 \norm{\xi_n - \omega^{\otimes n}}{1} \leq \frac12 \norm{\xi_n - \omega'_n}{1} + \frac12 \norm{\omega'_n - \omega^{\otimes n}}{1} \leq \ve.
\end{aligned}\end{equation}
On the other,
\begin{equation}\begin{aligned}
    \xi_n &\leq_{\FF} \frac{2^{-\eta} \sigma_n + 2^{\DD_s(\omega'_n\|\sigma_n)} \sigma_n}{1+2^{-\eta}} \\
    &= \frac{2^{-\eta} + 2^{\DD_s(\omega'_n\|\sigma_n)}}{1+2^{-\eta}}\, \sigma_n ,
\end{aligned}\end{equation}
and also
\begin{equation}\begin{aligned}
    \sigma_n &\leq \sigma_n + 2^{\eta} \omega'_n\\
    &= (1+2^{-\eta}) \frac{\sigma_n + 2^{\eta} \omega'_n}{1+2^{-\eta}}\\
    &= (2^\eta+1) \, \xi_n.
\end{aligned}\end{equation}
Putting all together,
\begin{equation}\begin{aligned}
    \frac1n\, \DD_{\Omega,s,\FF}(\xi_n) &\leq \frac1n \DD_{\Omega,s}(\xi_n\|\sigma_n) \\
    &= \frac1n \left( \DD_s(\xi_n\|\sigma_n) + D_{\max}(\sigma_n\|\xi_n) \right) \\
    &\leq \frac1n \left( \log \frac{2^{-\eta} + 2^{\DD_s(\omega'_n\|\sigma_n)}}{1+2^{-\eta}} + \log (2^\eta+1) \right) \\
    &= \frac1n \log \left( 1 + 2^\eta\, 2^{\DD_s(\omega'_n\|\sigma_n)} \right) \\
    &\leq \frac1n\, \DD_s(\omega'_n\|\sigma_n) + \frac1n \log(1+2^\eta).
\end{aligned}\end{equation}
Taking the limit then gives
\begin{equation}\begin{aligned}
    \limsup_{n\to\infty}\, \min_{\frac12 \norm{\omega'-\omega^{\otimes n}}{1}\leq \ve} \frac{1}{n} \DD_{\Omega,s,\FF} (\omega') &\leq \limsup_{n\to\infty}\, \frac{1}{n} \DD_{\Omega,s,\FF} (\xi_n) \\
    &\leq \limsup_{n\to\infty} \left(\frac1n\, \DD_s(\omega'_n\|\sigma_n) + \frac1n \log(1+2^\eta)\right) \\
    &= \limsup_{n\to\infty} \frac1n\, \DD_s(\omega'_n\|\sigma_n) \\
    &\leq \DD_{s,\FF}^\sminfty (\omega) + \delta.
\end{aligned}\end{equation}
Since the above holds for all sufficiently small $\ve>0$, we see that
\begin{equation}\begin{aligned}
    \DD_{\Omega,s,\FF}^\sminfty (\omega) &= \lim_{\ve\to 0} \limsup_{n\to\infty}\, \min_{\frac12 \norm{\omega'-\omega^{\otimes n}}{1}\leq \ve} \frac{1}{n} \DD_{\Omega,s,\FF} (\omega') \\
    &\leq \DD_{s,\FF}^\sminfty (\omega) + \delta.
\end{aligned}\end{equation}
As $\delta>0$ is arbitrary, this shows that
\begin{equation}\begin{aligned}
    \DD_{\Omega,s,\FF}^\sminfty (\omega) \leq \DD_{s,\FF}^\sminfty (\omega),
\end{aligned}\end{equation}
completing the proof.
\end{proof}

\begin{slemma}\label{lem:Rs2}
It holds that
\begin{equation}\begin{aligned}
  r_{\rm{prob}}(\rho \to \omega) \geq \frac{ \DD^\infty_{\Omega,\FF}(\rho)}{\DD_{\Omega,s,\FF}^\sminfty (\omega)}.
\end{aligned}\end{equation}
\end{slemma}

\begin{proof}
The proof of this statement is analogous to the proof of Proposition~\ref{prop:achievability_full} with $\DD^\sminfty_{\Omega,\FF}(\omega)$ replaced with $\DD^\sminfty_{\Omega,s,\FF}(\omega)$. Explicitly, fix $\delta > 0$ and consider any rate
\begin{equation}\begin{aligned}
  r = \frac{\DD^\infty_{\Omega,\FF}(\rho)}{\DD^\sminfty_{\Omega,s,\FF}(\omega)} \frac{1}{1+\delta}.
\end{aligned}\end{equation}
Let $(\omega_n)_n$ be a sequence of states such that $\frac12 \norm{\omega_n - \omega^{\otimes \floor{rn}}}{1} \leq \ve$ and such that
\begin{equation}\begin{aligned}
  \frac1n \DD_{\Omega,s,\FF_{\floor{rn}}}(\omega_n) &\leq (1+\delta) \lim_{\ve \to 0}\, \lim_{n\to\infty}\, \min_{\frac12\norm{\omega' - \omega^{\otimes \floor{rn}}}{1} \leq \ve}  \frac{1}{n} \DD_{\Omega,s,\FF_{\floor{rn}}} \left( \omega' \right)\\
&\leq (1+\delta) \, r \,\DD^\sminfty_{\Omega,s,\FF}(\omega)\\
&= \DD_{\Omega,\FF}^\infty(\rho),
\end{aligned}\end{equation}
where the second inequality follows from the weak sub-additivity of $\DD^\sminfty_{\Omega,s,\FF}$, specifically the fact that
\begin{equation}\begin{aligned}\label{eq:limsup_thingy}
  \lim_{\ve \to 0}\, \limsup_{n\to\infty}\, \min_{\frac12 \norm{\omega'-\omega^{\otimes \floor{rn}}}{1}\leq \ve} \frac{1}{n} \DD_{\Omega,s,\FF} (\omega') &= \lim_{\ve \to 0}\, \limsup_{n\to\infty}\, \frac{\floor{rn}}{n} \min_{\frac12 \norm{\omega'-\omega^{\otimes \floor{rn}}}{1}\leq \ve}  \frac{1}{\floor{rn}} \DD_{\Omega,s,\FF} (\omega')\\
  &\leq r \lim_{\ve \to 0}\, \limsup_{n\to\infty}\, \min_{\frac12 \norm{\omega'-\omega^{\otimes n}}{1}\leq \ve} \frac{1}{n} \DD_{\Omega,s,\FF} (\omega')
\end{aligned}\end{equation}
which follows from the definition of $\limsup$\footnote{The quantity $\DD_{\Omega,s,\FF}^\sminfty$ can actually be shown to be weakly additive and not merely sub-additive --- this can be proved as in the derivation of~\cite[Eq.~(41)]{brandao_2010-1}, but we do not need this fact here.}. This implies that
\begin{equation}\begin{aligned}
  \DD_{s,\FF_{\floor{rn}}}(\omega_n) \leq \DD_{\Omega,s,\FF_n}(\rho^{\otimes n}),
\end{aligned}\end{equation}
which means that the transformation from $\rho^{\otimes n}$ to $\omega_n$ can be realised probabilistically, up to an arbitrarily small error~\cite{regula_2021-4}.
\end{proof}

Lemmas~\ref{lem:Rs1} and~\ref{lem:Rs2} combined give that
\begin{equation}\begin{aligned}
r_{\rm{prob}}(\rho \to \omega) \geq \frac{ \DD^\infty_{\Omega,\FF}(\rho)}{\DD_{\Omega,s,\FF}^\sminfty (\omega)} = \frac{ \DD^\infty_{\Omega,\FF}(\rho)}{\DD_{s,\FF}^\sminfty (\omega)},
\end{aligned}\end{equation}
which is precisely the statement of Proposition~\ref{prop:achievability_Rs_full}.

%%%%%%%%%%%%%%%%%%%%%%%%%%%%%%%%%%%%%%%%%%%%%%%%%%%%%%%%%%%%%

\section{Achievability for distillation}\label{app:ach_dist}

We provide an alternative proof of an achievable bound for distillation that, although slightly weaker than Proposition~\ref{prop:achievability_Rs_full}, gives more insight into the transformation errors of the protocol. Here we employ the (non-smoothed) regulaized standard robustness
\begin{equation}\begin{aligned}
  \regrob(\rho) = \lim_{n\to\infty} \frac1n \log \left( 1 + R_{s,\F}(\rho^{\otimes n}) \right).
\end{aligned}\end{equation}

\begin{boxed}{white}
\begin{sproposition}\label{prop:dist_achiev}
Consider any resource theory in which $R_{s,\FF}(\rho) < \infty$ for every state $\rho$.
Then, for every pure target state $\psi$, it holds that
\begin{equation}\begin{aligned}
  r_{\mathrm{prob}}(\rho \to \psi) \geq \frac{ \DD_{\Omega,\FF}^\infty(\rho)}{\regrob(\psi)},
\end{aligned}\end{equation}
and the rate is achievable with error $\ve_n = 2^{-\Omega(n)}$.
\end{sproposition}
\end{boxed}

\begin{proof}
By~\cite[Thm.~12]{regula_2021-4}, for every $\ve_n$ such that
\begin{equation}\begin{aligned}\label{eq:dist_achiev_eq}
   R_{s,\FF_{\floor{rn}}}(\psi^{\otimes \floor{rn}}) \leq \left( \ve_n^{-1} - 1 \right)^{-1} \Omega_{\FF_n}(\rho^{\otimes n})
\end{aligned}\end{equation}
there exists a one-shot probabilistic protocol $\rho^{\otimes n} \to \psi^{\otimes \floor{rn}}$ with error (in fidelity) at most $\ve_n$. Any such protocol may %only succeed with a vanishingly small probability, 
consist of a sequence of free operations, but we can simply assume that for every $\delta_n > 0$, there exists a map $\Lambda_n \in \O$ such that $\frac{\Lambda_n(\rho^{\otimes})}{\Tr \Lambda_n(\rho^{\otimes n})} = \tau_n$ for some $\tau_n$ with $\braket{\psi^{\otimes \floor{rn}} | \tau_n | \psi^{\otimes \floor{rn}}} \geq 1-\ve_n - \delta_n$.

Let us fix some $\eta > 0$ and define $r = \frac{\DD_{\Omega,\FF}^\infty(\rho)}{\regrob(\psi)} - \eta$. 
Then, choosing
\begin{equation}\begin{aligned}
  \ve_n^{-1} = \frac{\Omega_{\FF_n}(\rho^{\otimes n})}{R_{s,\FF_{\floor{rn}}}(\psi^{\otimes \floor{rn}})} + 1
\end{aligned}\end{equation}
so that \eqref{eq:dist_achiev_eq} is satisfied, we ensure that there exists a protocol that takes $\rho$ to $\psi$ at a rate equal to $r$, with the transformation error for each $n$ given by $\ve_n + \delta_n$. Since the $\delta_n$ are arbitrary, it thus remains to show that $\ve_n \to 0$. This is ensured by the fact that
\begin{equation}\begin{aligned}
  \liminf_{n\to\infty} \ve_n^{-1} &\texteq{(i)} \liminf_{n\to\infty} \frac{\Omega_{\FF_n}(\rho^{\otimes n})}{R_{s,\FF_{\floor{rn}}}(\psi^{\otimes \floor{rn}}) + 1} + 1\\
  &\textgeq{(ii)} \liminf_{n\to\infty} \frac{\Omega_{\FF_n}(\rho^{\otimes n})}{ 2^{n r \regrob(\psi) + n \mu} } + 1\\
  &= \liminf_{n\to\infty} \frac{\Omega_{\FF_n}(\rho^{\otimes n})}{ 2^{n \DD_\Omega^\infty(\rho) - n \regrob(\psi) \eta + n \mu} } + 1\\
  &\textgeq{(iii)} \liminf_{n\to\infty} \frac{\Omega_{\FF_n}(\rho^{\otimes n})}{ \Omega_{\FF_n}(\rho^{\otimes n}) \, 2^{- n \regrob(\psi) \eta + n \mu} } + 1\\
  &= \liminf_{n\to\infty} \, 2^{n \left(\regrob(\psi) \eta - \mu\right)} + 1\\
  &\texteq{(iv)} \infty,
\end{aligned}\end{equation}
where in (i) we were free to add the constant term $+1$ in the denominator as it is irrelevant asymptotically, in (ii) we picked some $\mu > 0$ and took $n$ large enough so that
\begin{equation}\begin{aligned}
  \frac1n \log \left( 1 + R_{s,\FF_{\floor{rn}}}(\psi^{\otimes \floor{rn}}) \right) \leq r \, \regrob(\psi) + \mu,
\end{aligned}\end{equation}
in (iii) we used the fact that $\DD^\infty_{\Omega,\FF}(\rho) \leq \frac{1}{n} \log \RW(\rho^{\otimes n})$ for every $n$ by the subadditivity of $\DD_{\Omega,\FF}$, and in (iv) we observed that by picking $\mu < \regrob(\psi) \eta$ we can ensure that the term is unbounded. We thus have that the rate $\frac{\DD^\infty_{\Omega,\FF}(\rho)}{\regrob(\psi)} - \eta$ is achievable for all $\eta > 0$, and taking the supremum over such rates yields the stated result.
\end{proof}

It then follows immediately from the upper bound of Proposition~\ref{prop:dist_bound_rate_full} coupled with the achievability bounds of Propositions~\ref{prop:achievability_full} and \ref{prop:dist_achiev} that, when the given theory is affine and $D_{\min,\FF}^\infty(\psi) = D_{\FF}^\infty(\psi)$, 
or when $D_{\min,\FF}^\infty(\psi) = \regrob(\psi)$, then
\begin{equation}\begin{aligned}
  r_{\mathrm{prob}}(\rho \to \psi) = r^\dagger_{\mathrm{prob}}(\rho \to \psi) = \frac{ \DD_{\Omega,\FF}^\infty(\rho)}{D_{\min,\FF}^\infty(\psi)}.
\end{aligned}\end{equation}
An easier to verify condition for the upper and lower bounds to coincide is as follows.

\begin{boxed}{white}
\begin{scorollary}
Consider any state $\psi$ such that
\begin{enumerate}[(1)]
\item $\psi$ maximizes the max-relative entropy measure: for all $n$, $D_{\max,\FF_n}(\psi^{\otimes n})$ is maximal among all states of the same dimension. 
\item $D_{\min,\FF_n}(\psi^{\otimes n}) = n\, D_{\min,\FF}(\psi) \; \forall n$,
\end{enumerate}
and either
\begin{enumerate}
\item[(3a)] the given resource theory is affine, or
\item[(3b)] the %logarithm of the standard robustness
logarithmic standard robustness equals the max-relative entropy for $\psi^{\otimes n}$: specifically,
\begin{equation}\begin{aligned}
   D_{\max,\FF_n}(\psi^{\otimes n}) = \log \left( 1 + R_{s,\FF_n}(\psi^{\otimes n})\right) \; \forall n
 \end{aligned}\end{equation} (e.g.\ in the resource theory of entanglement).
\end{enumerate}
Then,
\begin{equation}\begin{aligned}
  r_{\mathrm{prob}}(\rho \to \psi) = r^\dagger_{\mathrm{prob}}(\rho \to \psi) = \frac{ \DD_{\Omega,\FF}^\infty(\rho)}{D_{\min,\FF}(\psi)}.
\end{aligned}\end{equation}
\end{scorollary}
\end{boxed}

\begin{proof}
The reason for introducing condition~(1) is perhaps not immediately clear: we do it because any state $\psi$ which maximises the max-relative entropy among all state of some dimension necessarily satisfies~\cite{regula_2020}
\begin{equation}\begin{aligned}\label{eq:golden}
  D_{\max,\FF}(\psi) = D_{\F}(\psi) = D_{\min,\FF}(\psi),
\end{aligned}\end{equation}
which is helpful in establishing an equality between our upper and lower bounds.

On the one hand, Proposition~\ref{prop:dist_bound_rate_full} ensures that
\begin{equation}\begin{aligned}
  r_{\rm prob}(\rho \to \psi) \leq r^\dagger_{\rm prob}(\rho \to \psi) \leq \frac{\DD^\infty_{\Omega,\FF} (\rho)}{D^\infty_{\min,\FF}(\psi)}.
\end{aligned}\end{equation}
On the other hand, condition~(3a) gives
\begin{equation}\begin{aligned}
  r_{\rm prob}(\rho \to \psi) \textgeq{(i)} \frac{\DD^\infty_{\Omega,\FF} (\rho)}{D^\infty_{\FF}(\psi)} \textgeq{(ii)} \frac{\DD^\infty_{\Omega,\FF} (\rho)}{D^\infty_{\max,\FF}(\psi)} \textgeq{(iii)} \frac{\DD^\infty_{\Omega,\FF} (\rho)}{D^\infty_{\min,\FF}(\psi)}
\end{aligned}\end{equation}
where we used: (i)~Proposition~\ref{prop:achievability_full}, (ii)~the fact that $D(\rho\|\sigma) \leq D_{\max}(\rho\|\sigma)$ for all states~\cite{datta_2009}, (iii)~condition~(1) and the ensuing Eq.~\eqref{eq:golden}.
Similarly, in the case of condition~(3b) we get
\begin{equation}\begin{aligned}
  r_{\rm prob}(\rho \to \psi) \textgeq{(i)} \frac{\DD^\infty_{\Omega,\FF} (\rho)}{\DD^\infty_{s,\FF}(\psi)} \texteq{(ii)} \frac{\DD^\infty_{\Omega,\FF} (\rho)}{D^\infty_{\max,\FF}(\psi)} \textgeq{(iii)} \frac{\DD^\infty_{\Omega,\FF} (\rho)}{D^\infty_{\min,\FF}(\psi)}
\end{aligned}\end{equation}
where we used: (i)~Proposition~\ref{prop:achievability_Rs_full}, (ii)~condition (3b), (iii)~Eq.~\eqref{eq:golden}. Invoking condition~(2) concludes the proof.
\end{proof}

We note that a state does not need to be a maximally resourceful one in order to satisfy $D_{\max,\FF}(\psi) = D_{\min,\FF}(\psi)$ (a counterexample being e.g.\ the Clifford magic states discussed in~\cite{bravyi_2019} in the resource theory of non-stabiliser quantum computation), but (1) is nevertheless a useful assumption to state as it can be satisfied in any given resource theory.

The additivity of the measures (condition (2)) and the equality between the standard and generalised robustness measures (condition (3b)) may be more difficult to ensure, depending on the given theory. We state them here because they are satisfied for some of the most important states of interest, and notably for the maximally entangled state in entanglement theory~\cite{vidal_1999,harrow_2003}.

%%%%%%%%%%%%%%%%%%%%%%%%%%%%%%%%%%%%%%%%%%%%%%%%%%%%%%%%%%%%%

\section{On the probability of success in the asymptotic limit}\label{sec:app_vanishing}

Recall that
\begin{equation}\begin{aligned}
  r_{\mathrm{prob} > 0}(\rho \to \omega) \coloneqq \sup \lsetr r\! \barr (\Lambda_n)_n \in \OO,\; \lim_{n \to \infty}  \frac12 \norm{\frac{\Lambda_n(\rho^{\otimes n})}{\Tr \Lambda_n(\rho^{\otimes n})} - \omega^{\otimes \floor{rn}}}{1} \!\!= 0,\; \liminf_{n \to \infty} \Tr \Lambda_n(\rho^{\otimes n}) >\! 0 \rsetr\!.
\end{aligned}\end{equation}

{
\renewcommand{\thetheorem}{\getrefnumber{prop:vanishing_success}}
\begin{boxed}{white}
\begin{proposition}
The rate of any state transformation with a probability of success that does not asymptotically vanish satisfies
\begin{equation}\begin{aligned} \label{eq:r_prob}
  r_{\rm{prob} >0}(\rho \to \omega) \leq \frac{D^\infty_\FF(\rho)}{D_\FF^\infty(\omega)}.
\end{aligned}\end{equation}
\end{proposition}
\end{boxed}
}

\begin{remark}
Before delving into the proof of the above result, it is instructive to try to use standard techniques (see e.g.~\cite{winter_2016} or~\cite[Eq.~(16)]{ferrari_2020}) and see exactly how and why they fail when applied to the probabilistic case. To this end, consider a sequence of transformations $(\Lambda_n)_n\in \OO$ with the property that $p_n\coloneqq \Tr \Lambda_n(\rho^{\otimes n})$ satisfies that $\liminf_{n\to\infty} p_n \eqqcolon p > 0$, and moreover $\omega_n \coloneqq \frac{1}{p_n}\, \Lambda(\rho^{\otimes n})$ has the property that $\ve_n \coloneqq \frac12 \norm{\omega_n - \omega^{\otimes \floor{rn}}}{1}$ obeys $\lim_{n\to\infty} \ve_n = 0$. Then one can pick an arbitrary sequence of free states $\sigma_n\in \FF$, and define the deterministic maps $\Lambda'_n(X) \coloneqq \Lambda_n(X) + (1-\Tr \Lambda_n(X))\sigma_n$. Since $\Lambda'_n\in \OO$ is a deterministic free operation for all $n$, we can try to apply the standard procedure to this new object. Doing so gives
\begin{equation}\begin{aligned} \label{eq:standard_techniques_fail}
D_\FF(\rho^{\otimes n}) &\geq D_\FF\left(\Lambda'_n(\rho^{\otimes n})\right) \\
&= D_\FF\left( p_n \omega_n + (1-p_n)\sigma_n\right) \\
&\geq p_n D_\FF\left( \omega_n \right) + (1-p_n)\, D_\FF\left( \sigma_n\right) - h_2(p_n) \\
&= p_n D_\FF\left( \omega_n \right) - h_2(p_n) \\
&\geq p_n \left[ D_\FF\left(\omega^{\otimes \floor{rn}}\right) - \floor{rn} \zeta \ve_n - (1+\ve_n) h_2\left(\frac{\ve_n}{1+\ve_n}\right) \right] - h_2(p_n) .
\end{aligned}\end{equation}
Here, the first line is by monotonicity of $D_\FF$, while the third descends from the observation that $D_\FF$ is `not too convex' (see e.g.~\cite[proof of Lemma~7]{winter_2016}), in turn an elementary consequence of the `not-too-concavity' of the von Neumann entropy, seemingly first established by Lanford and Robinson~\cite[Theorem~1]{Lanford68}. Finally, the last line of~\eqref{eq:standard_techniques_fail} is the most delicate. It follows from the asymptotic continuity of $D_\FF$ as established in~\cite[Lemma~7]{winter_2016}, which states that
\begin{equation}\begin{aligned}
\left|D_\FF(\rho) - D_\FF(\rho')\right| \leq \kappa \ve + (1+\ve)h_2\left(\frac{\ve}{1+\ve}\right)
\end{aligned}\end{equation}
for every pair of states $\rho,\rho'\in \D(\H)$ at trace distance $\ve \coloneqq \frac12 \norm{\rho-\rho'}{1}$, where $\kappa\coloneqq \sup_{\xi\in \D(\H)} D_\FF(\xi)$, the supremum being over all states $\xi$ on the same space $\H$ as $\rho$ and $\rho'$. In our case, since we know by Axiom~II in Section~\ref{sec:prelim} that there exists a full-rank free state $\sigma\geq \mu \id>0$ on the same space as $\omega$ such that $\sigma^{\otimes m}$ is also free for every $m$, calling $\H$ the Hilbert space pertaining to $\omega$ we see that
\begin{equation}\begin{aligned}
\kappa = \sup_{\xi\in \D(\H^{\otimes m})} D_\FF(\xi) \leq \sup_{\xi\in \D(\H^{\otimes m})} D(\xi\|\sigma^{\otimes m}) \leq -m \log \mu\, ,
\end{aligned}\end{equation}
which yields the desired inequality if one sets $m=\floor{rn}$ and $\zeta \coloneqq -\log \mu$.

Now, we can divide both sides of~\eqref{eq:standard_techniques_fail} by $n$ and take the limit $n\to\infty$. We obtain that
\begin{equation}\begin{aligned}
D_\FF^\infty(\rho) \geq p r D_\FF^\infty(\omega)\, ,
\end{aligned}\end{equation}
which does translate indeed into an upper bound on $r(\rho \to \omega)$, but still features an explicit dependence on $p$ and therefore yields no non-trivial upper bound on $r_{\rm{prob} >0}(\rho \to \omega)$.

The fundamental problem with the above technique is in the application of the not-too-concavity of $D_\FF$, which makes a coefficient $p_n$ appear in front of the relative entropy on the right-hand side. What we do below, instead, is substantially different: instead of leveraging the monotonicity of $D_\FF$ directly, we look at the projected relative entropy. As we saw, the key properties of this quantity is that, being invariant under rescaling, it can remove the explicit dependence on probabilities, and furthermore it is connected with the standard relative entropy of resource via our asymptotic equipartition property (Lemma~\ref{lem:omega_reg}).  
\end{remark}

\begin{proof}[Proof of Proposition~\ref{prop:vanishing_success}]
Assume that $r$ is an achievable rate for~\eqref{eq:r_prob}; that is, there exists a sequence of protocols $(\Lambda_n)_n \in \OO$ such that $\frac{\Lambda_n(\rho^{\otimes n})}{\Tr \Lambda_n(\rho^{\otimes n})} = \tau_n$ with $\frac12 \norm{\tau_n - \omega^{\otimes \floor{rn}}}{1} = \ve_n \to 0$, and furthermore $\liminf_{n\to \infty} \Tr \Lambda_n(\rho^{\otimes n}) \eqqcolon p > 0$. Define
\begin{equation}\begin{aligned}
  \DD_{\Omega,\FF}^\delta(\rho) \coloneqq \min_{\frac12 \norm{\rho-\rho'}{1} \leq \delta} \DD_{\Omega,\FF}(\rho'),
\end{aligned}\end{equation}
where the smoothing is over normalised density matrices $\rho'$. 
Defining the generalised trace distance as $D_1(X,Y) \coloneqq \frac12 \norm{X-Y}{1} + \frac12 \left|\Tr(X-Y)\right|$, this quantity satisfies that $D_1(\Lambda(\rho),\Lambda(\rho')) \leq D_1(\rho,\rho')$ under the action of any completely positive and trace non-increasing map $\Lambda$~\cite[Prop.~3.8]{tomamichel_2016}. We thus have that, for all states $\rho$ and $\rho'$, it holds that
\begin{equation}\begin{aligned}
  \frac12 \norm{\frac{\Lambda(\rho)}{\Tr \Lambda(\rho)} - \frac{\Lambda(\rho')}{\Tr \Lambda(\rho')} }{1} &\leq  \frac12 \norm{\frac{\Lambda(\rho)}{\Tr \Lambda(\rho)} - \frac{\Lambda(\rho')}{\Tr \Lambda(\rho)} }{1} +  \frac12 \norm{\frac{\Lambda(\rho')}{\Tr \Lambda(\rho)} - \frac{\Lambda(\rho')}{\Tr \Lambda(\rho')} }{1}\\
  &= \frac12 \norm{ \frac{ \Lambda(\rho) - \Lambda(\rho')}{\Tr \Lambda(\rho)}}{1} + \frac12 \norm{ \frac{\Lambda(\rho') \Tr\Lambda(\rho') - \Lambda(\rho') \Tr \Lambda(\rho)}{\Tr \Lambda(\rho) \Tr \Lambda(\rho')}}{1}\\
  &= \frac12 \norm{ \frac{ \Lambda(\rho) - \Lambda(\rho')}{\Tr \Lambda(\rho)}}{1} + \frac12 \left| \frac{\Tr \Lambda(\rho') - \Tr \Lambda(\rho)}{\Tr \Lambda(\rho)} \right|\\
  &\leq \frac{\frac12\norm{\rho-\rho'}{1}}{\Tr\Lambda(\rho)}.
\end{aligned}\end{equation}
From this, we immediately have that
\begin{equation}\begin{aligned}
  \DD_{\Omega,\FF}^\delta(\rho) \geq \DD_{\Omega,\FF}^{\delta/\Tr \Lambda(\rho)} \left(\frac{\Lambda(\rho)}{\Tr\Lambda(\rho)}\right)
\end{aligned}\end{equation}
for all free operations $\Lambda \in \OO$, using the monotonicity of $\DD_{\Omega,\FF}$.
Now, observe that for any $\delta$ such that $0<\delta < p$ it holds that
\begin{equation} \begin{aligned}
\limsup_{n\to\infty} \left(\frac{\delta}{\Tr\Lambda_n(\rho^{\otimes n})} + \ve_n\right) = \frac{\delta}{p} < 1\, .
\end{aligned}\end{equation}
Therefore, for all sufficiently large integers $n$ we have that $\delta/\Tr\Lambda_n(\rho^{\otimes n}) + \ve_n < 1$. Hence,
%Then, for any $\delta >0$ small enough and $n$ large enough so that $\delta/\Tr\Lambda_n(\rho^{\otimes n}) + \ve_n < 1$, we have that
\begin{equation}\begin{aligned}
  \DD_{\Omega,\FF}^\delta(\rho^{\otimes n}) &\geq \DD_{\Omega,\FF}^{\delta/\Tr\Lambda_n(\rho^{\otimes n})}(\tau_n)\\
  &\geq \min_{\frac12 \norm{\pi_n-\omega^{\otimes \floor{r n}}}{1}\leq \ve_n} \DD_{\Omega,\FF}^{\delta/\Tr\Lambda_n(\rho^{\otimes n})}(\pi_n)\\
  &\geq \DD_{\Omega,\FF}^{\delta/\Tr\Lambda_n(\rho^{\otimes n}) + \ve_n} (\omega^{\otimes \floor{rn}}),
\end{aligned}\end{equation}
where we used that for any state $\pi'_n$ such that $\frac12\norm{\pi'_n - \pi_n}{1} \leq \delta/\Tr\Lambda_n(\rho^{\otimes n})$ it holds that
\begin{equation}\begin{aligned}
  \frac12 \norm{\omega^{\otimes \floor{rn}} - \pi'_n}{1} \leq \frac12 \norm{\omega^{\otimes \floor{rn}} - \pi_n}{1} + \frac12 \norm{\pi_n - \pi'_n}{1} \leq \ve_n + %\delta/\Tr\Lambda_n(\rho^{\otimes n}).
  \frac{\delta}{\Tr\Lambda_n(\rho^{\otimes n})}.
\end{aligned}\end{equation}
This means that
\begin{equation}\begin{aligned}
  \lim_{\delta \to 0} \liminf_{n \to \infty} \frac1n \DD_{\Omega,\FF}^\delta(\rho^{\otimes n}) &\geq \lim_{\delta \to 0}  \liminf_{n \to \infty} \frac1n \DD_{\Omega,\FF}^{\delta/\Tr\Lambda_n(\rho^{\otimes n}) + \ve_n}(\omega^{\otimes \floor{rn}})\\
  & \geq \lim_{\delta \to 0} \, r\,  \liminf_{n \to \infty} \frac1n \DD_{\Omega,\FF}^{\delta/\Tr\Lambda_n(\rho^{\otimes n}) + \ve_n}(\omega^{\otimes n})\\
  & \geq r\, \lim_{\eta \to 0} \liminf_{n \to \infty} \frac1n \DD_{\Omega,\FF}^{\eta}(\omega^{\otimes n}),
\end{aligned}\end{equation}
where the second line follows analogously to \eqref{eq:limsup_thingy} and the last line follows since $\delta/\Tr\Lambda_n(\rho^{\otimes n}) + \ve_n \to 0$ as $n\to\infty$ and $\delta \to 0$.
Using Lemma~\ref{lem:omega_reg_full} gives the statement of the Proposition.
\end{proof}

%%%%%%%%%%%%%%%%%%%%%%%%%%%%%%%%%%%%%%%%%%%%%%%%%%%%%%%%%%%%%

\section{Isotropic states}\label{app:iso}

We consider a bipartite quantum system with two subsystems of dimension $d$. 
Below we will refer to two different sets of quantum states: the set of separable states $\SEP$, and the set of PPT (positive partial transpose) states $\PPT \coloneqq \lset \sigma \in \D(\H) \bar \sigma^\Gamma \in \D(\H) \rset$,  with $\Gamma$ denoting partial transposition over either of the two subsystems. It is well known that the partial transpose of any separable state is positive~\cite{horodecki_1996-1}.

Recall that isotropic states are defined as
\begin{equation}\begin{aligned}
  \rho_p = p \Phi_d + (1-p) \frac{\id - \Phi_d}{d^2-1}
\end{aligned}\end{equation}
and that $\rho_p \in \SEP \iff \rho_p \in \PPT \iff p \in \left[0, \frac1d\right]$~\cite{horodecki_1999-1}.

\enlargethispage{\baselineskip}
\begin{boxed}{white}
\begin{sproposition}\label{lem:isotropic}
For every isotropic state $\rho_p$ with local dimension $d$ and all $n \in \mathbb{N}$,
\begin{equation}\begin{aligned}
  \frac1n \DD_{\Omega,\SEP}(\rho_p^{\otimes n}) = \DD_{\Omega,\SEP}(\rho_p) = \begin{cases} \log \frac{p(d - 1)}{1 - p} & p \geq \frac1d,\\
  0 & p \leq \frac1d. \end{cases}
\end{aligned}\end{equation}
The same is true if $\SEP$ is replaced by the set $\PPT$.
\end{sproposition}
\end{boxed}
\begin{proof}
The separable case follows from the faithfulness of $\Omega_\SEP$, so assume that $p \geq 1/d$. 
The fact that $\Omega_\SEP(\rho_p) \leq \frac{p(d - 1)}{1 - p}$ can then be seen from the feasible solution
\begin{equation}\begin{aligned}
  \rho_p &\leq p d \,\rho_{1/d},\\
  \rho_{1/d} &\leq \frac{d-1}{(1-p) d} \,\rho_p.
\end{aligned}\end{equation}
Submultiplicativity of $\Omega_\SEP$ and the inclusion $\SEP \subseteq \PPT$ then immediately implies that
\begin{equation}\begin{aligned}\label{eq:iso_upper}
\Omega_{\PPT_n}(\rho_p^{\otimes n})^{1/n} \leq \Omega_{\SEP_n}(\rho_p^{\otimes n})^{1/n} \leq \frac{p(d - 1)}{1 - p}.
\end{aligned}\end{equation}

For the other inequality, consider the dual form of $\Omega_\PPT$ (and analogously $\Omega_\SEP$):
\begin{align}\label{eq:dual}
  \Omega_\PPT(\rho) &= \sup \lset \Tr A \rho \bar \Tr B \rho = 1,\; \Tr (B-A) \sigma \geq 0 \;\forall \sigma \in \PPT,\; A, B \geq 0\rset.
 \end{align}
 Constructing the feasible solutions
\begin{equation}\begin{aligned}
  A &= \frac{d-1}{1-p} \, \Phi_d,\\
  B &= \frac{\id - \Phi_d}{1-p}
\end{aligned}\end{equation}
we see that $\Tr A \rho_p = \frac{p(d - 1)}{1 - p}$ and $\Tr B \rho_p = 1$, so it remains to show that $B - A$ has a non-negative overlap with any PPT (or separable) state. Consider first that
\begin{equation}\begin{aligned}
  \max_{\sigma \in \PPT} \Tr A \sigma &= \max_{\sigma \in \PPT} \Tr A^\Gamma \sigma^\Gamma\\
  &\leq \max_{\rho \in \D(\H)} \Tr A^\Gamma \rho\\
  &= \lambda_{\max} \big(A^\Gamma\big)\\
  &= \frac{d-1}{1-p} \,\lambda_{\max}\left(\frac1d F\right)\\
  &= \frac{d-1}{d (1-p)},
  \end{aligned}\end{equation}
  where $\lambda_{\max}$ denotes the largest eigenvalue and $F$ is the swap operator, the eigenvalues of which are $\pm 1$. We also see that
\begin{equation}\begin{aligned}
   \min_{\sigma \in \PPT} \Tr B \sigma &= \min_{\sigma \in \PPT} B^\Gamma \sigma^\Gamma\\
   %&\geq \min_{\sigma \in \PPT} B^\Gamma \sigma^\Gamma\\
  &\geq \min_{\rho \in \D(\H)} \Tr B^\Gamma \rho\\
  &= \frac{1}{1-p} \lambda_{\min} \left(\id - \frac1d F\right)\\
  &= \frac{1}{1-p} \left(1-\frac1d\right),
\end{aligned}\end{equation}
where $\lambda_{\min}$ stands for the smallest eigenvalue. We thus have $\Tr B \sigma \geq \Tr A \sigma$ for every PPT state $\sigma$, and hence $\Omega_\PPT(\rho_p) \geq \frac{p(d - 1)}{1 - p}$ which means that equality must hold. To conclude the $n$-copy result, it suffices to notice that $A^{\otimes n}$ and $B^{\otimes n}$ are feasible solutions for $\rho_p^{\otimes n}$: the sufficiency of $A^{\otimes n}$ follows from the fact that $(A^{\otimes n})^\Gamma = (A^\Gamma)^{\otimes n}$ and the eigenvalues of $F^{\otimes n}$ are clearly $\pm 1$, and the sufficiency of $B^{\otimes n}$ follows analogously by noting that $\lambda_{\min}(P)$ is a multiplicative quantity for every positive semidefinite~$P$. This implies that
\begin{equation}\begin{aligned}
  \Omega_\SEP(\rho_p^{\otimes n}) \geq \Omega_\PPT(\rho_p^{\otimes n}) \geq \Tr \rho_p^{\otimes n} A^{\otimes n} = \left(\frac{p(d - 1)}{1 - p}\right)^n,
\end{aligned}\end{equation}
which together with~\eqref{eq:iso_upper} means that equality must hold in the above.
\end{proof}

%%%%%%%%%%%%%%%%%%%%%%%%%%%%%%%%%%%%%%%%%%%%%%%%%%%%%%%%%%%%%

 \bibliographystyle{apsc}
 \bibliography{main}

\begin{thebibliography}{65}%
\makeatletter
\providecommand \@ifxundefined [1]{%
 \@ifx{#1\undefined}
}%
\providecommand \@ifnum [1]{%
 \ifnum #1\expandafter \@firstoftwo
 \else \expandafter \@secondoftwo
 \fi
}%
\providecommand \@ifx [1]{%
 \ifx #1\expandafter \@firstoftwo
 \else \expandafter \@secondoftwo
 \fi
}%
\providecommand \natexlab [1]{#1}%
\providecommand \emph  [1]{``#1''}%
\providecommand \bibnamefont  [1]{#1}%
\providecommand \bibfnamefont [1]{#1}%
\providecommand \citenamefont [1]{#1}%
\providecommand \href@noop [0]{\@secondoftwo}%
\providecommand \href [0]{\begingroup \@sanitize@url \@href}%
\providecommand \@href[1]{\@@startlink{#1}\@@href}%
\providecommand \@@href[1]{\endgroup#1\@@endlink}%
\providecommand \@sanitize@url [0]{\catcode `\\12\catcode `\$12\catcode
  `\&12\catcode `\#12\catcode `\^12\catcode `\_12\catcode `\%12\relax}%
\providecommand \@@startlink[1]{}%
\providecommand \@@endlink[0]{}%
\providecommand \url  [0]{\begingroup\@sanitize@url \@url }%
\providecommand \@url [1]{\endgroup\@href {#1}{\urlprefix }}%
\providecommand \urlprefix  [0]{URL }%
\providecommand \Eprint [0]{\href }%
\providecommand \doibase [0]{http://dx.doi.org/}%
\providecommand \selectlanguage [0]{\@gobble}%
\providecommand \bibinfo  [0]{\@secondoftwo}%
\providecommand \bibfield  [0]{\@secondoftwo}%
\providecommand \translation [1]{[#1]}%
\providecommand \BibitemOpen [0]{}%
\providecommand \bibitemStop [0]{}%
\providecommand \bibitemNoStop [0]{.\EOS\space}%
\providecommand \EOS [0]{\spacefactor3000\relax}%
\providecommand \BibitemShut  [1]{\csname bibitem#1\endcsname}%
\let\auto@bib@innerbib\@empty
%</preamble>
\bibitem [{\citenamefont {Hayashi}(2006)}]{hayashi_2006-2}%
  \BibitemOpen
  \bibfield  {author} {\bibinfo {author} {\bibfnamefont {M.}~\bibnamefont
  {Hayashi}},\ }\href@noop {} {\emph {\bibinfo {title} {Quantum
  {{Information}}: {{An Introduction}}}}}\ (\bibinfo  {publisher} {{Springer
  Science \& Business Media}},\ \bibinfo {year} {2006})\BibitemShut {NoStop}%
\bibitem [{\citenamefont {Wilde}(2017)}]{wilde_2017}%
  \BibitemOpen
  \bibfield  {author} {\bibinfo {author} {\bibfnamefont {M.~M.}\ \bibnamefont
  {Wilde}},\ }\href@noop {} {\emph {\bibinfo {title} {Quantum Information
  Theory}}},\ \bibinfo {edition} {2nd}\ ed.\ (\bibinfo  {publisher} {Cambridge
  University Press},\ \bibinfo {year} {2017})\BibitemShut {NoStop}%
\bibitem [{\citenamefont {Hiai}\ and\ \citenamefont {Petz}(1991)}]{hiai_1991}%
  \BibitemOpen
  \bibfield  {author} {\bibinfo {author} {\bibfnamefont {F.}~\bibnamefont
  {Hiai}}\ and\ \bibinfo {author} {\bibfnamefont {D.}~\bibnamefont {Petz}},\
  }\bibfield  {title} {\emph {\bibinfo {title} {The proper formula for relative
  entropy and its asymptotics in quantum probability},}\ }\href
  {http://dx.doi.org/10.1007/BF02100287} {\bibfield  {journal} {\bibinfo
  {journal} {Commun. Math. Phys.}\ }\textbf {\bibinfo {volume} {143}},\
  \bibinfo {pages} {99} (\bibinfo {year} {1991})}\BibitemShut {NoStop}%
\bibitem [{\citenamefont {Ogawa}\ and\ \citenamefont
  {Nagaoka}(2000)}]{ogawa_2000}%
  \BibitemOpen
  \bibfield  {author} {\bibinfo {author} {\bibfnamefont {T.}~\bibnamefont
  {Ogawa}}\ and\ \bibinfo {author} {\bibfnamefont {H.}~\bibnamefont
  {Nagaoka}},\ }\bibfield  {title} {\emph {\bibinfo {title} {Strong converse
  and {{Stein}}'s lemma in quantum hypothesis testing},}\ }\href
  {http://dx.doi.org/10.1109/18.887855} {\bibfield  {journal} {\bibinfo
  {journal} {IEEE Trans. Inf. Theory}\ }\textbf {\bibinfo {volume} {46}},\
  \bibinfo {pages} {2428} (\bibinfo {year} {2000})}\BibitemShut {NoStop}%
\bibitem [{\citenamefont {Matsumoto}(2010)}]{matsumoto_2010}%
  \BibitemOpen
  \bibfield  {author} {\bibinfo {author} {\bibfnamefont {K.}~\bibnamefont
  {Matsumoto}},\ }\emph {\bibinfo {title} {Reverse {{Test}} and
  {{Characterization}} of {{Quantum Relative Entropy}}},}\ \href@noop {}
  {\Eprint {http://arxiv.org/abs/1010.1030} {arXiv:1010.1030}  (\bibinfo {year}
  {2010})}\BibitemShut {NoStop}%
\bibitem [{\citenamefont {Buscemi}\ \emph {et~al.}(2019)\citenamefont
  {Buscemi}, \citenamefont {Sutter},\ and\ \citenamefont
  {Tomamichel}}]{buscemi_2019}%
  \BibitemOpen
  \bibfield  {author} {\bibinfo {author} {\bibfnamefont {F.}~\bibnamefont
  {Buscemi}}, \bibinfo {author} {\bibfnamefont {D.}~\bibnamefont {Sutter}}, \
  and\ \bibinfo {author} {\bibfnamefont {M.}~\bibnamefont {Tomamichel}},\
  }\bibfield  {title} {\emph {\bibinfo {title} {An information-theoretic
  treatment of quantum dichotomies},}\ }\href
  {http://dx.doi.org/10.22331/q-2019-12-09-209} {\bibfield  {journal} {\bibinfo
   {journal} {Quantum}\ }\textbf {\bibinfo {volume} {3}},\ \bibinfo {pages}
  {209} (\bibinfo {year} {2019})}\BibitemShut {NoStop}%
\bibitem [{\citenamefont {Wang}\ and\ \citenamefont {Wilde}(2019)}]{wang_2019}%
  \BibitemOpen
  \bibfield  {author} {\bibinfo {author} {\bibfnamefont {X.}~\bibnamefont
  {Wang}}\ and\ \bibinfo {author} {\bibfnamefont {M.~M.}\ \bibnamefont
  {Wilde}},\ }\bibfield  {title} {\emph {\bibinfo {title} {Resource theory of
  asymmetric distinguishability},}\ }\href
  {http://dx.doi.org/10.1103/PhysRevResearch.1.033170} {\bibfield  {journal}
  {\bibinfo  {journal} {Phys. Rev. Research}\ }\textbf {\bibinfo {volume}
  {1}},\ \bibinfo {pages} {033170} (\bibinfo {year} {2019})}\BibitemShut
  {NoStop}%
\bibitem [{\citenamefont {Bennett}\ \emph {et~al.}(1996)\citenamefont
  {Bennett}, \citenamefont {Bernstein}, \citenamefont {Popescu},\ and\
  \citenamefont {Schumacher}}]{bennett_1996-1}%
  \BibitemOpen
  \bibfield  {author} {\bibinfo {author} {\bibfnamefont {C.~H.}\ \bibnamefont
  {Bennett}}, \bibinfo {author} {\bibfnamefont {H.~J.}\ \bibnamefont
  {Bernstein}}, \bibinfo {author} {\bibfnamefont {S.}~\bibnamefont {Popescu}},
  \ and\ \bibinfo {author} {\bibfnamefont {B.}~\bibnamefont {Schumacher}},\
  }\bibfield  {title} {\emph {\bibinfo {title} {Concentrating partial
  entanglement by local operations},}\ }\href
  {http://dx.doi.org/10.1103/PhysRevA.53.2046} {\bibfield  {journal} {\bibinfo
  {journal} {Phys. Rev. A}\ }\textbf {\bibinfo {volume} {53}},\ \bibinfo
  {pages} {2046} (\bibinfo {year} {1996})}\BibitemShut {NoStop}%
\bibitem [{\citenamefont {Brand{\~a}o}\ and\ \citenamefont
  {Plenio}(2010{\natexlab{a}})}]{brandao_2010}%
  \BibitemOpen
  \bibfield  {author} {\bibinfo {author} {\bibfnamefont {F.~G. S.~L.}\
  \bibnamefont {Brand{\~a}o}}\ and\ \bibinfo {author} {\bibfnamefont {M.~B.}\
  \bibnamefont {Plenio}},\ }\bibfield  {title} {\emph {\bibinfo {title} {A
  {{Reversible Theory}} of {{Entanglement}} and its {{Relation}} to the
  {{Second Law}}},}\ }\href {http://dx.doi.org/10.1007/s00220-010-1003-1}
  {\bibfield  {journal} {\bibinfo  {journal} {Commun. Math. Phys.}\ }\textbf
  {\bibinfo {volume} {295}},\ \bibinfo {pages} {829} (\bibinfo {year}
  {2010}{\natexlab{a}})}\BibitemShut {NoStop}%
\bibitem [{\citenamefont {Steinlechner}\ \emph {et~al.}(2012)\citenamefont
  {Steinlechner}, \citenamefont {Trojek}, \citenamefont {Jofre}, \citenamefont
  {Weier}, \citenamefont {Perez}, \citenamefont {Jennewein}, \citenamefont
  {Ursin}, \citenamefont {Rarity}, \citenamefont {Mitchell}, \citenamefont
  {Torres}, \citenamefont {Weinfurter},\ and\ \citenamefont
  {Pruneri}}]{steinlechner_12}%
  \BibitemOpen
  \bibfield  {author} {\bibinfo {author} {\bibfnamefont {F.}~\bibnamefont
  {Steinlechner}}, \bibinfo {author} {\bibfnamefont {P.}~\bibnamefont
  {Trojek}}, \bibinfo {author} {\bibfnamefont {M.}~\bibnamefont {Jofre}},
  \bibinfo {author} {\bibfnamefont {H.}~\bibnamefont {Weier}}, \bibinfo
  {author} {\bibfnamefont {D.}~\bibnamefont {Perez}}, \bibinfo {author}
  {\bibfnamefont {T.}~\bibnamefont {Jennewein}}, \bibinfo {author}
  {\bibfnamefont {R.}~\bibnamefont {Ursin}}, \bibinfo {author} {\bibfnamefont
  {J.}~\bibnamefont {Rarity}}, \bibinfo {author} {\bibfnamefont {M.~W.}\
  \bibnamefont {Mitchell}}, \bibinfo {author} {\bibfnamefont {J.~P.}\
  \bibnamefont {Torres}}, \bibinfo {author} {\bibfnamefont {H.}~\bibnamefont
  {Weinfurter}}, \ and\ \bibinfo {author} {\bibfnamefont {V.}~\bibnamefont
  {Pruneri}},\ }\bibfield  {title} {\emph {\bibinfo {title} {A high-brightness
  source of polarization-entangled photons optimized for applications in free
  space},}\ }\href {http://opg.optica.org/oe/abstract.cfm?URI=oe-20-9-9640}
  {\bibfield  {journal} {\bibinfo  {journal} {Opt. Express}\ }\textbf {\bibinfo
  {volume} {20}},\ \bibinfo {pages} {9640} (\bibinfo {year}
  {2012})}\BibitemShut {NoStop}%
\bibitem [{\citenamefont {Couteau}(2018)}]{couteau_2018}%
  \BibitemOpen
  \bibfield  {author} {\bibinfo {author} {\bibfnamefont {C.}~\bibnamefont
  {Couteau}},\ }\bibfield  {title} {\emph {\bibinfo {title} {Spontaneous
  parametric down-conversion},}\ }\href
  {http://dx.doi.org/10.1080/00107514.2018.1488463} {\bibfield  {journal}
  {\bibinfo  {journal} {Contemp. Phys.}\ }\textbf {\bibinfo {volume} {59}},\
  \bibinfo {pages} {291} (\bibinfo {year} {2018})}\BibitemShut {NoStop}%
\bibitem [{\citenamefont {Regula}(2022{\natexlab{a}})}]{regula_2022}%
  \BibitemOpen
  \bibfield  {author} {\bibinfo {author} {\bibfnamefont {B.}~\bibnamefont
  {Regula}},\ }\bibfield  {title} {\emph {\bibinfo {title} {Probabilistic
  {{Transformations}} of {{Quantum Resources}}},}\ }\href
  {http://dx.doi.org/10.1103/PhysRevLett.128.110505} {\bibfield  {journal}
  {\bibinfo  {journal} {Phys. Rev. Lett.}\ }\textbf {\bibinfo {volume} {128}},\
  \bibinfo {pages} {110505} (\bibinfo {year} {2022}{\natexlab{a}})}\BibitemShut
  {NoStop}%
\bibitem [{\citenamefont {Chitambar}\ and\ \citenamefont
  {Gour}(2019)}]{chitambar_2019}%
  \BibitemOpen
  \bibfield  {author} {\bibinfo {author} {\bibfnamefont {E.}~\bibnamefont
  {Chitambar}}\ and\ \bibinfo {author} {\bibfnamefont {G.}~\bibnamefont
  {Gour}},\ }\bibfield  {title} {\emph {\bibinfo {title} {Quantum resource
  theories},}\ }\href {http://dx.doi.org/10.1103/RevModPhys.91.025001}
  {\bibfield  {journal} {\bibinfo  {journal} {Rev. Mod. Phys.}\ }\textbf
  {\bibinfo {volume} {91}},\ \bibinfo {pages} {025001} (\bibinfo {year}
  {2019})}\BibitemShut {NoStop}%
\bibitem [{\citenamefont {Davies}\ and\ \citenamefont
  {Lewis}(1970)}]{davies_1970}%
  \BibitemOpen
  \bibfield  {author} {\bibinfo {author} {\bibfnamefont {E.~B.}\ \bibnamefont
  {Davies}}\ and\ \bibinfo {author} {\bibfnamefont {J.~T.}\ \bibnamefont
  {Lewis}},\ }\bibfield  {title} {\emph {\bibinfo {title} {An operational
  approach to quantum probability},}\ }\href
  {http://dx.doi.org/10.1007/BF01647093} {\bibfield  {journal} {\bibinfo
  {journal} {Commun. Math. Phys.}\ }\textbf {\bibinfo {volume} {17}},\ \bibinfo
  {pages} {239} (\bibinfo {year} {1970})}\BibitemShut {NoStop}%
\bibitem [{\citenamefont {Ozawa}(1984)}]{ozawa_1984}%
  \BibitemOpen
  \bibfield  {author} {\bibinfo {author} {\bibfnamefont {M.}~\bibnamefont
  {Ozawa}},\ }\bibfield  {title} {\emph {\bibinfo {title} {Quantum measuring
  processes of continuous observables},}\ }\href
  {http://dx.doi.org/10.1063/1.526000} {\bibfield  {journal} {\bibinfo
  {journal} {J. Math. Phys.}\ }\textbf {\bibinfo {volume} {25}},\ \bibinfo
  {pages} {79} (\bibinfo {year} {1984})}\BibitemShut {NoStop}%
\bibitem [{\citenamefont {Chitambar}\ \emph {et~al.}(2008)\citenamefont
  {Chitambar}, \citenamefont {Duan},\ and\ \citenamefont
  {Shi}}]{chitambar_2008}%
  \BibitemOpen
  \bibfield  {author} {\bibinfo {author} {\bibfnamefont {E.}~\bibnamefont
  {Chitambar}}, \bibinfo {author} {\bibfnamefont {R.}~\bibnamefont {Duan}}, \
  and\ \bibinfo {author} {\bibfnamefont {Y.}~\bibnamefont {Shi}},\ }\bibfield
  {title} {\emph {\bibinfo {title} {Tripartite {{Entanglement Transformations}}
  and {{Tensor Rank}}},}\ }\href
  {http://dx.doi.org/10.1103/PhysRevLett.101.140502} {\bibfield  {journal}
  {\bibinfo  {journal} {Phys. Rev. Lett.}\ }\textbf {\bibinfo {volume} {101}},\
  \bibinfo {pages} {140502} (\bibinfo {year} {2008})}\BibitemShut {NoStop}%
\bibitem [{\citenamefont {Yu}\ \emph {et~al.}(2014)\citenamefont {Yu},
  \citenamefont {Guo},\ and\ \citenamefont {Duan}}]{yu_2014}%
  \BibitemOpen
  \bibfield  {author} {\bibinfo {author} {\bibfnamefont {N.}~\bibnamefont
  {Yu}}, \bibinfo {author} {\bibfnamefont {C.}~\bibnamefont {Guo}}, \ and\
  \bibinfo {author} {\bibfnamefont {R.}~\bibnamefont {Duan}},\ }\bibfield
  {title} {\emph {\bibinfo {title} {Obtaining a $w$ state from a
  {Greenberger-Horne-Zeilinger} state via stochastic local operations and
  classical communication with a rate approaching unity},}\ }\href
  {https://link.aps.org/doi/10.1103/PhysRevLett.112.160401} {\bibfield
  {journal} {\bibinfo  {journal} {Phys. Rev. Lett.}\ }\textbf {\bibinfo
  {volume} {112}},\ \bibinfo {pages} {160401} (\bibinfo {year}
  {2014})}\BibitemShut {NoStop}%
\bibitem [{\citenamefont {Vrana}\ and\ \citenamefont
  {Christandl}(2017)}]{vrana_2017}%
  \BibitemOpen
  \bibfield  {author} {\bibinfo {author} {\bibfnamefont {P.}~\bibnamefont
  {Vrana}}\ and\ \bibinfo {author} {\bibfnamefont {M.}~\bibnamefont
  {Christandl}},\ }\bibfield  {title} {\emph {\bibinfo {title} {Entanglement
  {{Distillation}} from
  {{Greenberger}}\textendash{{Horne}}\textendash{{Zeilinger Shares}}},}\ }\href
  {http://dx.doi.org/10.1007/s00220-017-2861-6} {\bibfield  {journal} {\bibinfo
   {journal} {Commun. Math. Phys.}\ }\textbf {\bibinfo {volume} {352}},\
  \bibinfo {pages} {621} (\bibinfo {year} {2017})}\BibitemShut {NoStop}%
\bibitem [{\citenamefont {Fang}\ and\ \citenamefont {Liu}(2020)}]{fang_2020}%
  \BibitemOpen
  \bibfield  {author} {\bibinfo {author} {\bibfnamefont {K.}~\bibnamefont
  {Fang}}\ and\ \bibinfo {author} {\bibfnamefont {Z.-W.}\ \bibnamefont {Liu}},\
  }\bibfield  {title} {\emph {\bibinfo {title} {No-{{Go Theorems}} for
  {{Quantum Resource Purification}}},}\ }\href
  {http://dx.doi.org/10.1103/PhysRevLett.125.060405} {\bibfield  {journal}
  {\bibinfo  {journal} {Phys. Rev. Lett.}\ }\textbf {\bibinfo {volume} {125}},\
  \bibinfo {pages} {060405} (\bibinfo {year} {2020})}\BibitemShut {NoStop}%
\bibitem [{\citenamefont {Khatri}\ and\ \citenamefont
  {Wilde}(2020)}]{khatri_2020}%
  \BibitemOpen
  \bibfield  {author} {\bibinfo {author} {\bibfnamefont {S.}~\bibnamefont
  {Khatri}}\ and\ \bibinfo {author} {\bibfnamefont {M.~M.}\ \bibnamefont
  {Wilde}},\ }\bibfield  {title} {\emph {\bibinfo {title} {Principles of
  {{Quantum Communication Theory}}: {{A Modern Approach}}},}\ }\href@noop {}
  {\Eprint {http://arxiv.org/abs/2011.04672} {arXiv:2011.04672}  (\bibinfo
  {year} {2020})}\BibitemShut {NoStop}%
\bibitem [{\citenamefont {Umegaki}(1962)}]{umegaki_1962}%
  \BibitemOpen
  \bibfield  {author} {\bibinfo {author} {\bibfnamefont {H.}~\bibnamefont
  {Umegaki}},\ }\bibfield  {title} {\emph {\bibinfo {title} {Conditional
  expectation in an operator algebra, {{IV}} ({{Entropy}} and information)},}\
  }\href {http://dx.doi.org/10.2996/kmj/1138844604} {\bibfield  {journal}
  {\bibinfo  {journal} {Kodai Math. Sem. Rep.}\ }\textbf {\bibinfo {volume}
  {14}},\ \bibinfo {pages} {59} (\bibinfo {year} {1962})}\BibitemShut {NoStop}%
\bibitem [{\citenamefont {Vedral}\ and\ \citenamefont
  {Plenio}(1998)}]{vedral_1998}%
  \BibitemOpen
  \bibfield  {author} {\bibinfo {author} {\bibfnamefont {V.}~\bibnamefont
  {Vedral}}\ and\ \bibinfo {author} {\bibfnamefont {M.~B.}\ \bibnamefont
  {Plenio}},\ }\bibfield  {title} {\emph {\bibinfo {title} {Entanglement
  {{Measures}} and {{Purification Procedures}}},}\ }\href
  {http://dx.doi.org/10.1103/PhysRevA.57.1619} {\bibfield  {journal} {\bibinfo
  {journal} {Phys. Rev. A}\ }\textbf {\bibinfo {volume} {57}},\ \bibinfo
  {pages} {1619} (\bibinfo {year} {1998})}\BibitemShut {NoStop}%
\bibitem [{\citenamefont {Lami}\ and\ \citenamefont
  {Shirokov}(2021)}]{achievability}%
  \BibitemOpen
  \bibfield  {author} {\bibinfo {author} {\bibfnamefont {L.}~\bibnamefont
  {Lami}}\ and\ \bibinfo {author} {\bibfnamefont {M.~E.}\ \bibnamefont
  {Shirokov}},\ }\bibfield  {title} {\emph {\bibinfo {title} {Attainability and
  lower semi-continuity of the relative entropy of entanglement, and variations
  on the theme},}\ }\href {http://arxiv.org/abs/2105.08091} {\Eprint
  {http://arxiv.org/abs/2105.08091} {arXiv:2105.08091}  (\bibinfo {year}
  {2021})}\BibitemShut {NoStop}%
\bibitem [{\citenamefont {Vollbrecht}\ and\ \citenamefont
  {Werner}(2001)}]{vollbrecht_2001}%
  \BibitemOpen
  \bibfield  {author} {\bibinfo {author} {\bibfnamefont {K.~G.~H.}\
  \bibnamefont {Vollbrecht}}\ and\ \bibinfo {author} {\bibfnamefont {R.~F.}\
  \bibnamefont {Werner}},\ }\bibfield  {title} {\emph {\bibinfo {title}
  {Entanglement measures under symmetry},}\ }\href
  {http://dx.doi.org/10.1103/PhysRevA.64.062307} {\bibfield  {journal}
  {\bibinfo  {journal} {Phys. Rev. A}\ }\textbf {\bibinfo {volume} {64}},\
  \bibinfo {pages} {062307} (\bibinfo {year} {2001})}\BibitemShut {NoStop}%
\bibitem [{\citenamefont {Donald}\ \emph {et~al.}(2002)\citenamefont {Donald},
  \citenamefont {Horodecki},\ and\ \citenamefont {Rudolph}}]{donald_2002}%
  \BibitemOpen
  \bibfield  {author} {\bibinfo {author} {\bibfnamefont {M.~J.}\ \bibnamefont
  {Donald}}, \bibinfo {author} {\bibfnamefont {M.}~\bibnamefont {Horodecki}}, \
  and\ \bibinfo {author} {\bibfnamefont {O.}~\bibnamefont {Rudolph}},\
  }\bibfield  {title} {\emph {\bibinfo {title} {The uniqueness theorem for
  entanglement measures},}\ }\href {http://dx.doi.org/10.1063/1.1495917}
  {\bibfield  {journal} {\bibinfo  {journal} {J. Math. Phys.}\ }\textbf
  {\bibinfo {volume} {43}},\ \bibinfo {pages} {4252} (\bibinfo {year}
  {2002})}\BibitemShut {NoStop}%
\bibitem [{\citenamefont {Datta}(2009{\natexlab{a}})}]{datta_2009}%
  \BibitemOpen
  \bibfield  {author} {\bibinfo {author} {\bibfnamefont {N.}~\bibnamefont
  {Datta}},\ }\bibfield  {title} {\emph {\bibinfo {title} {Min- and
  {{Max-Relative Entropies}} and a {{New Entanglement Monotone}}},}\ }\href
  {http://dx.doi.org/10.1109/TIT.2009.2018325} {\bibfield  {journal} {\bibinfo
  {journal} {IEEE Trans. Inf. Theory}\ }\textbf {\bibinfo {volume} {55}},\
  \bibinfo {pages} {2816} (\bibinfo {year} {2009}{\natexlab{a}})}\BibitemShut
  {NoStop}%
\bibitem [{\citenamefont {Datta}(2009{\natexlab{b}})}]{datta_2009-2}%
  \BibitemOpen
  \bibfield  {author} {\bibinfo {author} {\bibfnamefont {N.}~\bibnamefont
  {Datta}},\ }\bibfield  {title} {\emph {\bibinfo {title} {Max-relative entropy
  of entanglement, alias log robustness},}\ }\href
  {http://dx.doi.org/10.1142/S0219749909005298} {\bibfield  {journal} {\bibinfo
   {journal} {Int. J. Quantum Inform.}\ }\textbf {\bibinfo {volume} {07}},\
  \bibinfo {pages} {475} (\bibinfo {year} {2009}{\natexlab{b}})}\BibitemShut
  {NoStop}%
\bibitem [{\citenamefont {Brand{\~a}o}\ and\ \citenamefont
  {Plenio}(2010{\natexlab{b}})}]{brandao_2010-1}%
  \BibitemOpen
  \bibfield  {author} {\bibinfo {author} {\bibfnamefont {F.~G. S.~L.}\
  \bibnamefont {Brand{\~a}o}}\ and\ \bibinfo {author} {\bibfnamefont {M.~B.}\
  \bibnamefont {Plenio}},\ }\bibfield  {title} {\emph {\bibinfo {title} {A
  {{Generalization}} of {{Quantum Stein}}'s {{Lemma}}},}\ }\href
  {http://dx.doi.org/10.1007/s00220-010-1005-z} {\bibfield  {journal} {\bibinfo
   {journal} {Commun. Math. Phys.}\ }\textbf {\bibinfo {volume} {295}},\
  \bibinfo {pages} {791} (\bibinfo {year} {2010}{\natexlab{b}})}\BibitemShut
  {NoStop}%
\bibitem [{\citenamefont {Berta}\ \emph {et~al.}(2022)\citenamefont {Berta},
  \citenamefont {Brand{\~a}o}, \citenamefont {Gour}, \citenamefont {Lami},
  \citenamefont {Plenio}, \citenamefont {Regula},\ and\ \citenamefont
  {Tomamichel}}]{berta_2022}%
  \BibitemOpen
  \bibfield  {author} {\bibinfo {author} {\bibfnamefont {M.}~\bibnamefont
  {Berta}}, \bibinfo {author} {\bibfnamefont {F.~G. S.~L.}\ \bibnamefont
  {Brand{\~a}o}}, \bibinfo {author} {\bibfnamefont {G.}~\bibnamefont {Gour}},
  \bibinfo {author} {\bibfnamefont {L.}~\bibnamefont {Lami}}, \bibinfo {author}
  {\bibfnamefont {M.~B.}\ \bibnamefont {Plenio}}, \bibinfo {author}
  {\bibfnamefont {B.}~\bibnamefont {Regula}}, \ and\ \bibinfo {author}
  {\bibfnamefont {M.}~\bibnamefont {Tomamichel}},\ }\bibfield  {title} {\emph
  {\bibinfo {title} {On a gap in the proof of the generalised quantum
  {{Stein}}'s lemma and its consequences for the reversibility of quantum
  resources},}\ }\href {http://arxiv.org/abs/2205.02813} {\Eprint
  {http://arxiv.org/abs/2205.02813} {arXiv:2205.02813}  (\bibinfo {year}
  {2022})}\BibitemShut {NoStop}%
\bibitem [{\citenamefont {Bushell}(1973)}]{bushell_1973}%
  \BibitemOpen
  \bibfield  {author} {\bibinfo {author} {\bibfnamefont {P.~J.}\ \bibnamefont
  {Bushell}},\ }\bibfield  {title} {\emph {\bibinfo {title} {Hilbert's metric
  and positive contraction mappings in a {{Banach}} space},}\ }\href
  {http://dx.doi.org/10.1007/BF00247467} {\bibfield  {journal} {\bibinfo
  {journal} {Arch. Rat. Mech. Anal.}\ }\textbf {\bibinfo {volume} {52}},\
  \bibinfo {pages} {330} (\bibinfo {year} {1973})}\BibitemShut {NoStop}%
\bibitem [{\citenamefont {Reeb}\ \emph {et~al.}(2011)\citenamefont {Reeb},
  \citenamefont {Kastoryano},\ and\ \citenamefont {Wolf}}]{reeb_2011}%
  \BibitemOpen
  \bibfield  {author} {\bibinfo {author} {\bibfnamefont {D.}~\bibnamefont
  {Reeb}}, \bibinfo {author} {\bibfnamefont {M.~J.}\ \bibnamefont
  {Kastoryano}}, \ and\ \bibinfo {author} {\bibfnamefont {M.~M.}\ \bibnamefont
  {Wolf}},\ }\bibfield  {title} {\emph {\bibinfo {title} {Hilbert's projective
  metric in quantum information theory},}\ }\href
  {http://dx.doi.org/10.1063/1.3615729} {\bibfield  {journal} {\bibinfo
  {journal} {J. Math. Phys.}\ }\textbf {\bibinfo {volume} {52}},\ \bibinfo
  {pages} {082201} (\bibinfo {year} {2011})}\BibitemShut {NoStop}%
\bibitem [{\citenamefont {Regula}(2022{\natexlab{b}})}]{regula_2021-4}%
  \BibitemOpen
  \bibfield  {author} {\bibinfo {author} {\bibfnamefont {B.}~\bibnamefont
  {Regula}},\ }\bibfield  {title} {\emph {\bibinfo {title} {Tight constraints
  on probabilistic convertibility of quantum states},}\ }\href
  {http://dx.doi.org/10.22331/q-2022-09-22-817} {\bibfield  {journal} {\bibinfo
   {journal} {Quantum}\ }\textbf {\bibinfo {volume} {6}},\ \bibinfo {pages}
  {817} (\bibinfo {year} {2022}{\natexlab{b}})}\BibitemShut {NoStop}%
\bibitem [{\citenamefont {Horodecki}(2001)}]{horodecki_2001}%
  \BibitemOpen
  \bibfield  {author} {\bibinfo {author} {\bibfnamefont {M.}~\bibnamefont
  {Horodecki}},\ }\bibfield  {title} {\emph {\bibinfo {title} {Entanglement
  {{Measures}}},}\ }\href {http://dl.acm.org/citation.cfm?id=2011326.2011328}
  {\bibfield  {journal} {\bibinfo  {journal} {Quant. Inf. Comput.}\ }\textbf
  {\bibinfo {volume} {1}},\ \bibinfo {pages} {3} (\bibinfo {year}
  {2001})}\BibitemShut {NoStop}%
\bibitem [{\citenamefont {Horodecki}\ and\ \citenamefont
  {Oppenheim}(2013)}]{horodecki_2013-3}%
  \BibitemOpen
  \bibfield  {author} {\bibinfo {author} {\bibfnamefont {M.}~\bibnamefont
  {Horodecki}}\ and\ \bibinfo {author} {\bibfnamefont {J.}~\bibnamefont
  {Oppenheim}},\ }\bibfield  {title} {\emph {\bibinfo {title} {({Q}uantumness
  in the context of) {R}esource theories},}\ }\href
  {http://dx.doi.org/10.1142/S0217979213450197} {\bibfield  {journal} {\bibinfo
   {journal} {Int. J. Mod. Phys. B}\ }\textbf {\bibinfo {volume} {27}},\
  \bibinfo {pages} {1345019} (\bibinfo {year} {2013})}\BibitemShut {NoStop}%
\bibitem [{\citenamefont {Gour}(2017)}]{gour_2017}%
  \BibitemOpen
  \bibfield  {author} {\bibinfo {author} {\bibfnamefont {G.}~\bibnamefont
  {Gour}},\ }\bibfield  {title} {\emph {\bibinfo {title} {Quantum resource
  theories in the single-shot regime},}\ }\href
  {http://dx.doi.org/10.1103/PhysRevA.95.062314} {\bibfield  {journal}
  {\bibinfo  {journal} {Phys. Rev. A}\ }\textbf {\bibinfo {volume} {95}},\
  \bibinfo {pages} {062314} (\bibinfo {year} {2017})}\BibitemShut {NoStop}%
\bibitem [{\citenamefont {Regula}\ \emph {et~al.}(2020)\citenamefont {Regula},
  \citenamefont {Bu}, \citenamefont {Takagi},\ and\ \citenamefont
  {Liu}}]{regula_2020}%
  \BibitemOpen
  \bibfield  {author} {\bibinfo {author} {\bibfnamefont {B.}~\bibnamefont
  {Regula}}, \bibinfo {author} {\bibfnamefont {K.}~\bibnamefont {Bu}}, \bibinfo
  {author} {\bibfnamefont {R.}~\bibnamefont {Takagi}}, \ and\ \bibinfo {author}
  {\bibfnamefont {Z.-W.}\ \bibnamefont {Liu}},\ }\bibfield  {title} {\emph
  {\bibinfo {title} {Benchmarking one-shot distillation in general quantum
  resource theories},}\ }\href {http://dx.doi.org/10.1103/PhysRevA.101.062315}
  {\bibfield  {journal} {\bibinfo  {journal} {Phys. Rev. A}\ }\textbf {\bibinfo
  {volume} {101}},\ \bibinfo {pages} {062315} (\bibinfo {year}
  {2020})}\BibitemShut {NoStop}%
\bibitem [{\citenamefont {Brand{\~a}o}\ \emph {et~al.}(2013)\citenamefont
  {Brand{\~a}o}, \citenamefont {Horodecki}, \citenamefont {Oppenheim},
  \citenamefont {Renes},\ and\ \citenamefont {Spekkens}}]{brandao_2013}%
  \BibitemOpen
  \bibfield  {author} {\bibinfo {author} {\bibfnamefont {F.~G. S.~L.}\
  \bibnamefont {Brand{\~a}o}}, \bibinfo {author} {\bibfnamefont
  {M.}~\bibnamefont {Horodecki}}, \bibinfo {author} {\bibfnamefont
  {J.}~\bibnamefont {Oppenheim}}, \bibinfo {author} {\bibfnamefont {J.~M.}\
  \bibnamefont {Renes}}, \ and\ \bibinfo {author} {\bibfnamefont {R.~W.}\
  \bibnamefont {Spekkens}},\ }\bibfield  {title} {\emph {\bibinfo {title}
  {Resource {{Theory}} of {{Quantum States Out}} of {{Thermal Equilibrium}}},}\
  }\href {http://dx.doi.org/10.1103/PhysRevLett.111.250404} {\bibfield
  {journal} {\bibinfo  {journal} {Phys. Rev. Lett.}\ }\textbf {\bibinfo
  {volume} {111}},\ \bibinfo {pages} {250404} (\bibinfo {year}
  {2013})}\BibitemShut {NoStop}%
\bibitem [{\citenamefont {Baumgratz}\ \emph {et~al.}(2014)\citenamefont
  {Baumgratz}, \citenamefont {Cramer},\ and\ \citenamefont
  {Plenio}}]{baumgratz_2014}%
  \BibitemOpen
  \bibfield  {author} {\bibinfo {author} {\bibfnamefont {T.}~\bibnamefont
  {Baumgratz}}, \bibinfo {author} {\bibfnamefont {M.}~\bibnamefont {Cramer}}, \
  and\ \bibinfo {author} {\bibfnamefont {M.~B.}\ \bibnamefont {Plenio}},\
  }\bibfield  {title} {\emph {\bibinfo {title} {Quantifying {{Coherence}}},}\
  }\href {http://dx.doi.org/10.1103/PhysRevLett.113.140401} {\bibfield
  {journal} {\bibinfo  {journal} {Phys. Rev. Lett.}\ }\textbf {\bibinfo
  {volume} {113}},\ \bibinfo {pages} {140401} (\bibinfo {year}
  {2014})}\BibitemShut {NoStop}%
\bibitem [{\citenamefont {Gour}\ and\ \citenamefont
  {Spekkens}(2008)}]{gour_2008}%
  \BibitemOpen
  \bibfield  {author} {\bibinfo {author} {\bibfnamefont {G.}~\bibnamefont
  {Gour}}\ and\ \bibinfo {author} {\bibfnamefont {R.~W.}\ \bibnamefont
  {Spekkens}},\ }\bibfield  {title} {\emph {\bibinfo {title} {The resource
  theory of quantum reference frames: Manipulations and monotones},}\ }\href
  {http://dx.doi.org/10.1088/1367-2630/10/3/033023} {\bibfield  {journal}
  {\bibinfo  {journal} {New J. Phys.}\ }\textbf {\bibinfo {volume} {10}},\
  \bibinfo {pages} {033023} (\bibinfo {year} {2008})}\BibitemShut {NoStop}%
\bibitem [{\citenamefont {Wu}\ \emph {et~al.}(2021)\citenamefont {Wu},
  \citenamefont {Kondra}, \citenamefont {Rana}, \citenamefont {Scandolo},
  \citenamefont {Xiang}, \citenamefont {Li}, \citenamefont {Guo},\ and\
  \citenamefont {Streltsov}}]{wu_2021}%
  \BibitemOpen
  \bibfield  {author} {\bibinfo {author} {\bibfnamefont {K.-D.}\ \bibnamefont
  {Wu}}, \bibinfo {author} {\bibfnamefont {T.~V.}\ \bibnamefont {Kondra}},
  \bibinfo {author} {\bibfnamefont {S.}~\bibnamefont {Rana}}, \bibinfo {author}
  {\bibfnamefont {C.~M.}\ \bibnamefont {Scandolo}}, \bibinfo {author}
  {\bibfnamefont {G.-Y.}\ \bibnamefont {Xiang}}, \bibinfo {author}
  {\bibfnamefont {C.-F.}\ \bibnamefont {Li}}, \bibinfo {author} {\bibfnamefont
  {G.-C.}\ \bibnamefont {Guo}}, \ and\ \bibinfo {author} {\bibfnamefont
  {A.}~\bibnamefont {Streltsov}},\ }\bibfield  {title} {\emph {\bibinfo {title}
  {Operational {{Resource Theory}} of {{Imaginarity}}},}\ }\href
  {http://dx.doi.org/10.1103/PhysRevLett.126.090401} {\bibfield  {journal}
  {\bibinfo  {journal} {Phys. Rev. Lett.}\ }\textbf {\bibinfo {volume} {126}},\
  \bibinfo {pages} {090401} (\bibinfo {year} {2021})}\BibitemShut {NoStop}%
\bibitem [{\citenamefont {Faist}\ \emph {et~al.}(2019)\citenamefont {Faist},
  \citenamefont {Berta},\ and\ \citenamefont {Brand{\~a}o}}]{faist_2019}%
  \BibitemOpen
  \bibfield  {author} {\bibinfo {author} {\bibfnamefont {P.}~\bibnamefont
  {Faist}}, \bibinfo {author} {\bibfnamefont {M.}~\bibnamefont {Berta}}, \ and\
  \bibinfo {author} {\bibfnamefont {F.}~\bibnamefont {Brand{\~a}o}},\
  }\bibfield  {title} {\emph {\bibinfo {title} {Thermodynamic {{Capacity}} of
  {{Quantum Processes}}},}\ }\href
  {http://dx.doi.org/10.1103/PhysRevLett.122.200601} {\bibfield  {journal}
  {\bibinfo  {journal} {Phys. Rev. Lett.}\ }\textbf {\bibinfo {volume} {122}},\
  \bibinfo {pages} {200601} (\bibinfo {year} {2019})}\BibitemShut {NoStop}%
\bibitem [{\citenamefont {Chitambar}(2018)}]{chitambar_2018}%
  \BibitemOpen
  \bibfield  {author} {\bibinfo {author} {\bibfnamefont {E.}~\bibnamefont
  {Chitambar}},\ }\bibfield  {title} {\emph {\bibinfo {title}
  {Dephasing-covariant operations enable asymptotic reversibility of quantum
  resources},}\ }\href {http://dx.doi.org/10.1103/PhysRevA.97.050301}
  {\bibfield  {journal} {\bibinfo  {journal} {Phys. Rev. A}\ }\textbf {\bibinfo
  {volume} {97}},\ \bibinfo {pages} {050301} (\bibinfo {year}
  {2018})}\BibitemShut {NoStop}%
\bibitem [{\citenamefont {Brand{\~a}o}\ and\ \citenamefont
  {Gour}(2015)}]{brandao_2015}%
  \BibitemOpen
  \bibfield  {author} {\bibinfo {author} {\bibfnamefont {F.~G. S.~L.}\
  \bibnamefont {Brand{\~a}o}}\ and\ \bibinfo {author} {\bibfnamefont
  {G.}~\bibnamefont {Gour}},\ }\bibfield  {title} {\emph {\bibinfo {title}
  {Reversible framework for quantum resource theories},}\ }\href
  {http://dx.doi.org/10.1103/PhysRevLett.115.070503} {\bibfield  {journal}
  {\bibinfo  {journal} {Phys. Rev. Lett.}\ }\textbf {\bibinfo {volume} {115}},\
  \bibinfo {pages} {070503} (\bibinfo {year} {2015})}\BibitemShut {NoStop}%
\bibitem [{\citenamefont {Lami}\ and\ \citenamefont
  {Regula}(2023)}]{lami_2021-1}%
  \BibitemOpen
  \bibfield  {author} {\bibinfo {author} {\bibfnamefont {L.}~\bibnamefont
  {Lami}}\ and\ \bibinfo {author} {\bibfnamefont {B.}~\bibnamefont {Regula}},\
  }\bibfield  {title} {\emph {\bibinfo {title} {No second law of entanglement
  manipulation after all},}\ }\href
  {http://dx.doi.org/10.1038/s41567-022-01873-9} {\bibfield  {journal}
  {\bibinfo  {journal} {Nat. Phys.}\ }\textbf {\bibinfo {volume} {19}},\
  \bibinfo {pages} {184} (\bibinfo {year} {2023})}\BibitemShut {NoStop}%
\bibitem [{\citenamefont {Vidal}\ and\ \citenamefont
  {Tarrach}(1999)}]{vidal_1999}%
  \BibitemOpen
  \bibfield  {author} {\bibinfo {author} {\bibfnamefont {G.}~\bibnamefont
  {Vidal}}\ and\ \bibinfo {author} {\bibfnamefont {R.}~\bibnamefont
  {Tarrach}},\ }\bibfield  {title} {\emph {\bibinfo {title} {Robustness of
  entanglement},}\ }\href {http://dx.doi.org/10.1103/PhysRevA.59.141}
  {\bibfield  {journal} {\bibinfo  {journal} {Phys. Rev. A}\ }\textbf {\bibinfo
  {volume} {59}},\ \bibinfo {pages} {141} (\bibinfo {year} {1999})}\BibitemShut
  {NoStop}%
\bibitem [{\citenamefont {Rains}(1999{\natexlab{a}})}]{rains_1999}%
  \BibitemOpen
  \bibfield  {author} {\bibinfo {author} {\bibfnamefont {E.~M.}\ \bibnamefont
  {Rains}},\ }\bibfield  {title} {\emph {\bibinfo {title} {Rigorous treatment
  of distillable entanglement},}\ }\href
  {http://dx.doi.org/10.1103/PhysRevA.60.173} {\bibfield  {journal} {\bibinfo
  {journal} {Phys. Rev. A}\ }\textbf {\bibinfo {volume} {60}},\ \bibinfo
  {pages} {173} (\bibinfo {year} {1999}{\natexlab{a}})}\BibitemShut {NoStop}%
\bibitem [{\citenamefont {Winter}\ and\ \citenamefont
  {Yang}(2016)}]{winter_2016}%
  \BibitemOpen
  \bibfield  {author} {\bibinfo {author} {\bibfnamefont {A.}~\bibnamefont
  {Winter}}\ and\ \bibinfo {author} {\bibfnamefont {D.}~\bibnamefont {Yang}},\
  }\bibfield  {title} {\emph {\bibinfo {title} {Operational resource theory of
  coherence},}\ }\href {http://dx.doi.org/10.1103/PhysRevLett.116.120404}
  {\bibfield  {journal} {\bibinfo  {journal} {Phys. Rev. Lett.}\ }\textbf
  {\bibinfo {volume} {116}},\ \bibinfo {pages} {120404} (\bibinfo {year}
  {2016})}\BibitemShut {NoStop}%
\bibitem [{\citenamefont {{Contreras-Tejada}}\ \emph
  {et~al.}(2019)\citenamefont {{Contreras-Tejada}}, \citenamefont
  {Palazuelos},\ and\ \citenamefont {{de Vicente}}}]{contreras-tejada_2019}%
  \BibitemOpen
  \bibfield  {author} {\bibinfo {author} {\bibfnamefont {P.}~\bibnamefont
  {{Contreras-Tejada}}}, \bibinfo {author} {\bibfnamefont {C.}~\bibnamefont
  {Palazuelos}}, \ and\ \bibinfo {author} {\bibfnamefont {J.~I.}\ \bibnamefont
  {{de Vicente}}},\ }\bibfield  {title} {\emph {\bibinfo {title} {Resource
  {{Theory}} of {{Entanglement}} with a {{Unique Multipartite Maximally
  Entangled State}}},}\ }\href
  {http://dx.doi.org/10.1103/PhysRevLett.122.120503} {\bibfield  {journal}
  {\bibinfo  {journal} {Phys. Rev. Lett.}\ }\textbf {\bibinfo {volume} {122}},\
  \bibinfo {pages} {120503} (\bibinfo {year} {2019})}\BibitemShut {NoStop}%
\bibitem [{\citenamefont {Horodecki}\ \emph {et~al.}(1999)\citenamefont
  {Horodecki}, \citenamefont {Horodecki},\ and\ \citenamefont
  {Horodecki}}]{horodecki_1999-1}%
  \BibitemOpen
  \bibfield  {author} {\bibinfo {author} {\bibfnamefont {M.}~\bibnamefont
  {Horodecki}}, \bibinfo {author} {\bibfnamefont {P.}~\bibnamefont
  {Horodecki}}, \ and\ \bibinfo {author} {\bibfnamefont {R.}~\bibnamefont
  {Horodecki}},\ }\bibfield  {title} {\emph {\bibinfo {title} {General
  teleportation channel, singlet fraction, and quasidistillation},}\ }\href
  {http://dx.doi.org/10.1103/PhysRevA.60.1888} {\bibfield  {journal} {\bibinfo
  {journal} {Phys. Rev. A}\ }\textbf {\bibinfo {volume} {60}},\ \bibinfo
  {pages} {1888} (\bibinfo {year} {1999})}\BibitemShut {NoStop}%
\bibitem [{\citenamefont {Rains}(1999{\natexlab{b}})}]{rains_1999-1}%
  \BibitemOpen
  \bibfield  {author} {\bibinfo {author} {\bibfnamefont {E.~M.}\ \bibnamefont
  {Rains}},\ }\bibfield  {title} {\emph {\bibinfo {title} {Bound on distillable
  entanglement},}\ }\href {http://dx.doi.org/10.1103/PhysRevA.60.179}
  {\bibfield  {journal} {\bibinfo  {journal} {Phys. Rev. A}\ }\textbf {\bibinfo
  {volume} {60}},\ \bibinfo {pages} {179} (\bibinfo {year}
  {1999}{\natexlab{b}})}\BibitemShut {NoStop}%
\bibitem [{\citenamefont {Regula}\ \emph {et~al.}(2022)\citenamefont {Regula},
  \citenamefont {Lami},\ and\ \citenamefont {Wilde}}]{regula_2022-4}%
  \BibitemOpen
  \bibfield  {author} {\bibinfo {author} {\bibfnamefont {B.}~\bibnamefont
  {Regula}}, \bibinfo {author} {\bibfnamefont {L.}~\bibnamefont {Lami}}, \ and\
  \bibinfo {author} {\bibfnamefont {M.~M.}\ \bibnamefont {Wilde}},\ }\bibfield
  {title} {\emph {\bibinfo {title} {Postselected quantum hypothesis testing},}\
  }\href@noop {} {\Eprint {http://arxiv.org/abs/2209.10550} {arXiv:2209.10550}
  (\bibinfo {year} {2022})}\BibitemShut {NoStop}%
\bibitem [{\citenamefont {Gour}\ \emph {et~al.}(2015)\citenamefont {Gour},
  \citenamefont {M{\"u}ller}, \citenamefont {Narasimhachar}, \citenamefont
  {Spekkens},\ and\ \citenamefont {Yunger~Halpern}}]{gour_2015}%
  \BibitemOpen
  \bibfield  {author} {\bibinfo {author} {\bibfnamefont {G.}~\bibnamefont
  {Gour}}, \bibinfo {author} {\bibfnamefont {M.~P.}\ \bibnamefont
  {M{\"u}ller}}, \bibinfo {author} {\bibfnamefont {V.}~\bibnamefont
  {Narasimhachar}}, \bibinfo {author} {\bibfnamefont {R.~W.}\ \bibnamefont
  {Spekkens}}, \ and\ \bibinfo {author} {\bibfnamefont {N.}~\bibnamefont
  {Yunger~Halpern}},\ }\bibfield  {title} {\emph {\bibinfo {title} {The
  resource theory of informational nonequilibrium in thermodynamics},}\ }\href
  {http://dx.doi.org/10.1016/j.physrep.2015.04.003} {\bibfield  {journal}
  {\bibinfo  {journal} {Phys. Rep.}\ }\textbf {\bibinfo {volume} {583}},\
  \bibinfo {pages} {1} (\bibinfo {year} {2015})}\BibitemShut {NoStop}%
\bibitem [{\citenamefont {Boyd}\ and\ \citenamefont
  {Vandenberghe}(2004)}]{boyd_2004}%
  \BibitemOpen
  \bibfield  {author} {\bibinfo {author} {\bibfnamefont {S.}~\bibnamefont
  {Boyd}}\ and\ \bibinfo {author} {\bibfnamefont {L.}~\bibnamefont
  {Vandenberghe}},\ }\href@noop {} {\emph {\bibinfo {title} {Convex
  {{Optimization}}}}}\ (\bibinfo  {publisher} {{Cambridge University Press}},\
  \bibinfo {address} {{New York}},\ \bibinfo {year} {2004})\BibitemShut
  {NoStop}%
\bibitem [{\citenamefont {Gurvits}\ and\ \citenamefont
  {Barnum}(2002)}]{gurvits_2002}%
  \BibitemOpen
  \bibfield  {author} {\bibinfo {author} {\bibfnamefont {L.}~\bibnamefont
  {Gurvits}}\ and\ \bibinfo {author} {\bibfnamefont {H.}~\bibnamefont
  {Barnum}},\ }\bibfield  {title} {\emph {\bibinfo {title} {Largest separable
  balls around the maximally mixed bipartite quantum state},}\ }\href
  {http://dx.doi.org/10.1103/PhysRevA.66.062311} {\bibfield  {journal}
  {\bibinfo  {journal} {Phys. Rev. A}\ }\textbf {\bibinfo {volume} {66}},\
  \bibinfo {pages} {062311} (\bibinfo {year} {2002})}\BibitemShut {NoStop}%
\bibitem [{\citenamefont {Fekete}(1923)}]{fekete_1923}%
  \BibitemOpen
  \bibfield  {author} {\bibinfo {author} {\bibfnamefont {M.}~\bibnamefont
  {Fekete}},\ }\bibfield  {title} {\emph {\bibinfo {title} {{\"Uber die
  Verteilung der Wurzeln bei gewissen algebraischen Gleichungen mit
  ganzzahligen Koeffizienten}},}\ }\href {http://dx.doi.org/10.1007/BF01504345}
  {\bibfield  {journal} {\bibinfo  {journal} {Math. Z.}\ }\textbf {\bibinfo
  {volume} {17}},\ \bibinfo {pages} {228} (\bibinfo {year} {1923})}\BibitemShut
  {NoStop}%
\bibitem [{\citenamefont {Tomamichel}\ and\ \citenamefont
  {Hayashi}(2013)}]{tomamichel_2013}%
  \BibitemOpen
  \bibfield  {author} {\bibinfo {author} {\bibfnamefont {M.}~\bibnamefont
  {Tomamichel}}\ and\ \bibinfo {author} {\bibfnamefont {M.}~\bibnamefont
  {Hayashi}},\ }\bibfield  {title} {\emph {\bibinfo {title} {A {{Hierarchy}} of
  {{Information Quantities}} for {{Finite Block Length Analysis}} of {{Quantum
  Tasks}}},}\ }\href {http://dx.doi.org/10.1109/TIT.2013.2276628} {\bibfield
  {journal} {\bibinfo  {journal} {IEEE Trans. Inf. Theory}\ }\textbf {\bibinfo
  {volume} {59}},\ \bibinfo {pages} {7693} (\bibinfo {year}
  {2013})}\BibitemShut {NoStop}%
\bibitem [{\citenamefont {Datta}\ \emph {et~al.}(2013)\citenamefont {Datta},
  \citenamefont {Mosonyi}, \citenamefont {Hsieh},\ and\ \citenamefont
  {Brand{\~a}o}}]{datta_2013-1}%
  \BibitemOpen
  \bibfield  {author} {\bibinfo {author} {\bibfnamefont {N.}~\bibnamefont
  {Datta}}, \bibinfo {author} {\bibfnamefont {M.}~\bibnamefont {Mosonyi}},
  \bibinfo {author} {\bibfnamefont {M.~H.}\ \bibnamefont {Hsieh}}, \ and\
  \bibinfo {author} {\bibfnamefont {F.~G. S.~L.}\ \bibnamefont {Brand{\~a}o}},\
  }\bibfield  {title} {\emph {\bibinfo {title} {A {{Smooth Entropy Approach}}
  to {{Quantum Hypothesis Testing}} and the {{Classical Capacity}} of {{Quantum
  Channels}}},}\ }\href {http://dx.doi.org/10.1109/TIT.2013.2282160} {\bibfield
   {journal} {\bibinfo  {journal} {IEEE Trans. Inf. Theory}\ }\textbf {\bibinfo
  {volume} {59}},\ \bibinfo {pages} {8014} (\bibinfo {year}
  {2013})}\BibitemShut {NoStop}%
\bibitem [{\citenamefont {Gour}\ \emph {et~al.}(2009)\citenamefont {Gour},
  \citenamefont {Marvian},\ and\ \citenamefont {Spekkens}}]{gour_2009}%
  \BibitemOpen
  \bibfield  {author} {\bibinfo {author} {\bibfnamefont {G.}~\bibnamefont
  {Gour}}, \bibinfo {author} {\bibfnamefont {I.}~\bibnamefont {Marvian}}, \
  and\ \bibinfo {author} {\bibfnamefont {R.~W.}\ \bibnamefont {Spekkens}},\
  }\bibfield  {title} {\emph {\bibinfo {title} {Measuring the quality of a
  quantum reference frame: {{The}} relative entropy of frameness},}\ }\href
  {http://dx.doi.org/10.1103/PhysRevA.80.012307} {\bibfield  {journal}
  {\bibinfo  {journal} {Phys. Rev. A}\ }\textbf {\bibinfo {volume} {80}},\
  \bibinfo {pages} {012307} (\bibinfo {year} {2009})}\BibitemShut {NoStop}%
\bibitem [{\citenamefont {Piani}(2009)}]{piani_2009-1}%
  \BibitemOpen
  \bibfield  {author} {\bibinfo {author} {\bibfnamefont {M.}~\bibnamefont
  {Piani}},\ }\bibfield  {title} {\emph {\bibinfo {title} {Relative {{Entropy}}
  of {{Entanglement}} and {{Restricted Measurements}}},}\ }\href
  {http://dx.doi.org/10.1103/PhysRevLett.103.160504} {\bibfield  {journal}
  {\bibinfo  {journal} {Phys. Rev. Lett.}\ }\textbf {\bibinfo {volume} {103}},\
  \bibinfo {pages} {160504} (\bibinfo {year} {2009})}\BibitemShut {NoStop}%
\bibitem [{\citenamefont {Bravyi}\ \emph {et~al.}(2019)\citenamefont {Bravyi},
  \citenamefont {Browne}, \citenamefont {Calpin}, \citenamefont {Campbell},
  \citenamefont {Gosset},\ and\ \citenamefont {Howard}}]{bravyi_2019}%
  \BibitemOpen
  \bibfield  {author} {\bibinfo {author} {\bibfnamefont {S.}~\bibnamefont
  {Bravyi}}, \bibinfo {author} {\bibfnamefont {D.}~\bibnamefont {Browne}},
  \bibinfo {author} {\bibfnamefont {P.}~\bibnamefont {Calpin}}, \bibinfo
  {author} {\bibfnamefont {E.}~\bibnamefont {Campbell}}, \bibinfo {author}
  {\bibfnamefont {D.}~\bibnamefont {Gosset}}, \ and\ \bibinfo {author}
  {\bibfnamefont {M.}~\bibnamefont {Howard}},\ }\bibfield  {title} {\emph
  {\bibinfo {title} {Simulation of quantum circuits by low-rank stabilizer
  decompositions},}\ }\href {http://dx.doi.org/10.22331/q-2019-09-02-181}
  {\bibfield  {journal} {\bibinfo  {journal} {Quantum}\ }\textbf {\bibinfo
  {volume} {3}},\ \bibinfo {pages} {181} (\bibinfo {year} {2019})}\BibitemShut
  {NoStop}%
\bibitem [{\citenamefont {Harrow}\ and\ \citenamefont
  {Nielsen}(2003)}]{harrow_2003}%
  \BibitemOpen
  \bibfield  {author} {\bibinfo {author} {\bibfnamefont {A.~W.}\ \bibnamefont
  {Harrow}}\ and\ \bibinfo {author} {\bibfnamefont {M.~A.}\ \bibnamefont
  {Nielsen}},\ }\bibfield  {title} {\emph {\bibinfo {title} {Robustness of
  quantum gates in the presence of noise},}\ }\href
  {http://dx.doi.org/10.1103/PhysRevA.68.012308} {\bibfield  {journal}
  {\bibinfo  {journal} {Phys. Rev. A}\ }\textbf {\bibinfo {volume} {68}},\
  \bibinfo {pages} {012308} (\bibinfo {year} {2003})}\BibitemShut {NoStop}%
\bibitem [{\citenamefont {Ferrari}\ \emph {et~al.}(2020)\citenamefont
  {Ferrari}, \citenamefont {Lami}, \citenamefont {Theurer},\ and\ \citenamefont
  {Plenio}}]{ferrari_2020}%
  \BibitemOpen
  \bibfield  {author} {\bibinfo {author} {\bibfnamefont {G.}~\bibnamefont
  {Ferrari}}, \bibinfo {author} {\bibfnamefont {L.}~\bibnamefont {Lami}},
  \bibinfo {author} {\bibfnamefont {T.}~\bibnamefont {Theurer}}, \ and\
  \bibinfo {author} {\bibfnamefont {M.~B.}\ \bibnamefont {Plenio}},\ }\bibfield
   {title} {\emph {\bibinfo {title} {Asymptotic state transformations of
  continuous variable resources},}\ }\href {http://arxiv.org/abs/2010.00044}
  {\Eprint {http://arxiv.org/abs/2010.00044} {arXiv:2010.00044}  (\bibinfo
  {year} {2020})}\BibitemShut {NoStop}%
\bibitem [{\citenamefont {Lanford}\ and\ \citenamefont
  {Robinson}(1968)}]{Lanford68}%
  \BibitemOpen
  \bibfield  {author} {\bibinfo {author} {\bibfnamefont {O.~E.}\ \bibnamefont
  {Lanford}}\ and\ \bibinfo {author} {\bibfnamefont {D.~W.}\ \bibnamefont
  {Robinson}},\ }\bibfield  {title} {\emph {\bibinfo {title} {Mean entropy of
  states in quantum‐statistical mechanics},}\ }\href
  {https://doi.org/10.1063/1.1664685} {\bibfield  {journal} {\bibinfo
  {journal} {J. Math. Phys.}\ }\textbf {\bibinfo {volume} {9}},\ \bibinfo
  {pages} {1120} (\bibinfo {year} {1968})}\BibitemShut {NoStop}%
\bibitem [{\citenamefont {Tomamichel}(2016)}]{tomamichel_2016}%
  \BibitemOpen
  \bibfield  {author} {\bibinfo {author} {\bibfnamefont {M.}~\bibnamefont
  {Tomamichel}},\ }\href {//www.springer.com/gp/book/9783319218908} {\emph
  {\bibinfo {title} {Quantum {{Information Processing}} with {{Finite
  Resources}}}}}\ (\bibinfo  {publisher} {{Springer}},\ \bibinfo {year}
  {2016})\BibitemShut {NoStop}%
\bibitem [{\citenamefont {Horodecki}\ \emph {et~al.}(1996)\citenamefont
  {Horodecki}, \citenamefont {Horodecki},\ and\ \citenamefont
  {Horodecki}}]{horodecki_1996-1}%
  \BibitemOpen
  \bibfield  {author} {\bibinfo {author} {\bibfnamefont {M.}~\bibnamefont
  {Horodecki}}, \bibinfo {author} {\bibfnamefont {P.}~\bibnamefont
  {Horodecki}}, \ and\ \bibinfo {author} {\bibfnamefont {R.}~\bibnamefont
  {Horodecki}},\ }\bibfield  {title} {\emph {\bibinfo {title} {Separability of
  mixed states: Necessary and sufficient conditions},}\ }\href
  {http://dx.doi.org/10.1016/S0375-9601(96)00706-2} {\bibfield  {journal}
  {\bibinfo  {journal} {Phys. Lett. A}\ }\textbf {\bibinfo {volume} {223}},\
  \bibinfo {pages} {1} (\bibinfo {year} {1996})}\BibitemShut {NoStop}%
\end{thebibliography}%
 
%%%%%%%%%%%%%%%%%%%%%%%%%%%%%%%%%%%%%%%%%%%%%%%%%%%%%%%%%%%%%
\end{document}